\documentclass[sigconf]{acmart}
\usepackage{tabularx}
\usepackage{array}
\usepackage{booktabs}
\usepackage{amsmath} % 或 mathtools
\usepackage{enumitem}
\usepackage{subcaption}
\usepackage{multirow}   % 如果以后需要跨行
\usepackage{xcolor}
\usepackage{balance}
\usepackage{mathtools}
\usepackage{pifont}
\usepackage{setspace}
\usepackage{enumitem}
\AtBeginDocument{%
  }
\usepackage[ruled,vlined,linesnumbered]{algorithm2e}

\setcopyright{none}
\acmDOI{}
\acmISBN{}
\copyrightyear{2026}
\acmYear{2026}
\acmConference[Preprint]{Preprint}{2026}{}
\newcommand{\yin}[1]{\textcolor{magenta}{#1}}
\long\def\comment#1{}

\begin{document}

%%
%% The "title" command has an optional parameter,
%% allowing the author to define a "short title" to be used in page headers.
\title{Efficient Graph Indexing for Interval-Aware Vector Search}

%%
%% The "author" command and its associated commands are used to define
%% the authors and their affiliations.
%% Of note is the shared affiliation of the first two authors, and the
%% "authornote" and "authornotemark" commands
%% used to denote shared contribution to the research.
\settopmatter{authorsperrow=4}

\author{Siyuan Liang}
\affiliation{%
  \institution{Beijing Institute of Technology}
  \city{Beijing}
  \country{China}}
\email{lsy2004bit@126.com}

\author{Ziqi Yin}
\affiliation{%
  \institution{Nanyang Technological University}
  \city{Singapore}
  \country{Singapore}}
\email{ziqi003@e.ntu.edu.sg}

\author{Qi Zhang}
\affiliation{%
  \institution{University of Science and Technology Beijing}
  \city{Beijing}
  \country{China}}
\email{qizhangcs@ustb.edu.cn}

\author{Ronghua Li}
\affiliation{%
  \institution{Beijing Institute of Technology}
  \city{Beijing}
  \country{China}}
\email{lironghuabit@126.com}

\author{Guoren Wang}
\affiliation{%
  \institution{Beijing Institute of Technology}
  \city{Beijing}
  \country{China}}
\email{wanggrbit@126.com}

\author{Kaiwen Xue}
\affiliation{%
  \institution{Huawei Technologies Co., Ltd.}
  \city{Beijing}
  \country{China}}
\email{xuekaiwen6@huawei.com}

\author{Daiyin Wang}
\affiliation{%
  \institution{Huawei Technologies Co., Ltd.}
  \city{Beijing}
  \country{China}}
\email{wangdaiyin@huawei.com}

\author{Xubin Li}
\affiliation{%
  \institution{Huawei Technologies Co., Ltd.}
  \city{Beijing}
  \country{China}}
\email{lixubin@huawei.com}

\renewcommand{\shortauthors}{Liang et al.}

%%
%% The abstract is a short summary of the work to be presented in the
%% article.liang
%  

\begin{abstract}
Interval-aware Approximate Nearest Neighbor (ANN) search arises in applications where each object is associated with a numeric value or interval, and queries must satisfy both vector-similarity and interval constraints. Existing methods are typically tailored to a single query semantics, such as interval-filtered ANN search, and therefore require multiple specialized indexes to support diverse workloads, leading to substantial indexing and memory overhead. To address this limitation, we propose the Unified Interval-aware Relative Neighborhood Graph (URNG), a unified graph framework for interval-aware ANN search. URNG preserves the monotonic searchability of relative-neighborhood-graph based ANN indexes while additionally ensuring structural heredity over query-induced subgraphs, enabling a single index to support multiple interval-aware query semantics. Building on this framework, we develop UG, a practical graph index that efficiently approximates URNG through unified interval-aware pruning and iterative repair, together with a query algorithm for interval-aware ANN search. Extensive experiments on 5 datasets show that UG consistently achieves a strong accuracy-efficiency trade-off across diverse interval-aware workloads while maintaining competitive index construction cost and memory usage.

\end{abstract}

\maketitle

\section{Introduction}

Driven by advances in deep learning and large language models~\cite{lecun2015deep,brown2020language,reimers2019sentence}, vector search has become a core component in many real-world applications, including recommendation systems~\cite{cover1967nearest,sarwar2001item}, information retrieval~\cite{liu2007survey}, data mining~\cite{ester1996density,breunig2000lof}, and retrieval-augmented generation (RAG)~\cite{zhao2026retrieval}. However, due to the curse of dimensionality~\cite{beyer1999nearest,indyk1998approximate,weber1998quantitative}, exact nearest neighbor (NN) search in high-dimensional spaces is often too expensive to satisfy practical latency requirements. Approximate nearest neighbor (ANN) search addresses this challenge by substantially improving efficiency with only limited loss in accuracy.

Over the past few decades, a wide range of ANN indexes have been proposed, including tree-based, hashing-based, quantization-based, and graph-based methods~\cite{arora2018hd,beygelzimer2006cover,wang2014hashing,wang2017survey,gao2024rabitq,ge2013optimized,fu2021high,fu2017fast,yin2025deg,yin2026bbcimprovinglargekapproximate}. Among them, graph-based methods consistently achieve a strong trade-off between query accuracy and efficiency~\cite{azizi2025graph,wang2021comprehensive}. Their effectiveness fundamentally relies on graph navigability: well-connected routing structures enable the search to approach high-quality neighbors while examining only a small fraction of the dataset. In particular, graph indexes that approximate the Relative Neighborhood Graph (RNG) are especially appealing, as they often examine fewer objects while maintaining high recall~\cite{toussaint1980relative}. 

% These methods organize data objects as graph nodes and connect nearby vectors through proximity-aware edges. During query processing, search typically starts from an entry node and progressively moves toward the query through greedy graph traversal, often with beam-style exploration. 

However, many real-world similarity search tasks are not purely vector-based. In addition to vector similarity, retrieved objects are often required to satisfy interval constraints associated with numeric attributes such as time, price, or validity ranges. A representative example arises in video surveillance and object tracking, where each detected object is associated not only with an embedding vector, but also with a temporal interval indicating when it appears in the camera's field of view. Depending on the application, the system may need to support different interval-aware ANN query, as illustrated below.
\begin{enumerate}[leftmargin=*,topsep=0pt]
\item Given a query vector $v$ and a time window $q.I$, one may want to retrieve visually similar objects whose visible durations are fully contained by the query range, %, i.e., $[o.a_s,o.a_t]\subseteq q.I$. This
which is referred to as Interval-Filtering ANN (IFANN) query~\cite{yang2025hi}. %Range-Filtering ANN (RFANN) queries can be regarded as a special case of IFANN, where each object is associated with a scalar attribute $o.a$ rather than an interval, namely $o.a_s=o.a_t=o.a$, and the filtering condition becomes $o.a\in q.I$~\cite{zou2026rnsgrangeawaregraphindex,xu2024irangegraph}.

\item Conversely, given a query vector $v$ and a time interval $q.I$, one may want to retrieve visually similar objects whose visible durations cover the entire query interval, %i.e., $[o.a_s,o.a_t]\supseteq q.I$. This
which is referred to as Interval-Stabbing ANN (ISANN) queries. When the query interval degenerates to a single timestamp $t$, %i.e., $q.I=[t,t]$, '
ISANN reduces to Range-Stabbing ANN (RSANN) queries, also known as the timestamp ANN query~\cite{wang2025timestamp}, which retrieve nearest neighbors that remain active at the specific time $t$.
\end{enumerate}
%For example, given a query vector $v$ and a time window, one may want to retrieve the nearest neighbors whose temporal intervals are fully contained in the query range, namely Interval-Filtering ANN (IFANN) queries~\cite{yang2025hi}. In another case, given a query vector and a timestamp, one may instead retrieve the nearest neighbors that are active at that specific time, namely Range-Stabbing ANN (RSANN) queries~\cite{wang2025timestamp}.

A natural question is whether a single graph index can support these diverse interval-aware ANN queries efficiently. Unfortunately, this is non-trivial. The key challenge is structural: under interval constraints, search is effectively restricted to the query-valid objects, inducing a query-dependent subgraph of the original graph. Even if the original graph is highly navigable, the induced subgraph may lose critical routing nodes and edges, thereby severely degrading search quality. Moreover, different query semantics, such as interval containment and timestamp stabbing, can induce different subgraphs under the same range condition. As a result, a graph index that performs well for standard ANN search does not necessarily remain effective for interval-aware ANN retrieval.

Partly due to this challenge, existing interval-aware ANN methods are typically specialized to a single query semantics rather than designed as a unified index. For example, Hi-PNG~\cite{yang2025hi} targets IFANN queries by building multiple graph indexes over different sub-ranges and selecting qualified sub-indexes for routing at query time. While effective for its target setting, such methods do not directly support a unified family of interval-aware queries. Consequently, supporting multiple query semantics often requires maintaining multiple specialized indexes, which leads to substantial indexing and memory overhead.

In this paper, we propose the Unified interval-aware Relative Neighborhood Graph (URNG), a unified graph framework for interval-aware ANN search under numeric inclusion semantics. Unlike existing filtered graph designs, URNG explicitly accounts for both encompassing and subsumed relationships during graph construction. By carefully designing the pruning conditions, URNG preserves the search-friendly structure of RNG-style graphs under interval-aware filtering, while incurring only a modest constant-factor overhead in storage under a uniform interval model. These properties enable URNG to support a broad family of interval inclusion queries within a single graph index.

Constructing an exact URNG on large-scale datasets, however, remains computationally prohibitive. To bridge theory and practice, we further develop the Unified Graph index (UG), a practical graph index that efficiently approximates URNG. The construction of UG consists of three phases: (1) generating an initial candidate set for each node by jointly considering spatial and interval proximity; (2) constructing a local UG for each node via iterative pruning; and (3) merging these local structures to approximate the global graph. Based on UG, we further design an efficient query processing algorithm for interval-aware ANN search. In summary, our main contributions are as follows:

\comment{
\yin{To support these diverse interval–aware ANN queries, a variety of indexes have been proposed. For example, ~\cite{yang2025hi} proposes Hi-PNG to handle Interval-Filtered ANN queries, which constructs multiple graph-based indexes for different sub-ranges and selects the sub-indexes that satisfy the query conditions for routing at query time. However, they are all designed to handle a specific type of queries, but not for all. This results in substantial overhead in terms of indexing time and memory consumption. %For instance, xx. There is an approach that can support all~[xx], but it cannot support any best.
}}

\begin{itemize}[leftmargin=*,topsep=0pt]
    \item {\textbf{Unified graph framework.} We propose the URNG, a unified theoretical framework for interval-aware ANN search under %numeric inclusion
    various interval semantics queries. URNG preserves the monotonic searchability of RNG-style graphs while ensuring structural heredity over query-induced subgraphs, i.e., the subgraph induced by a valid query range remains a valid URNG. We further show that, under a uniform interval model, URNG incurs only a constant-factor overhead over the corresponding RNG in index size and search complexity.}
% These properties enable URNG to support a broad family of interval inclusion queries within a single graph index.
    \item {\textbf{Practical unified index.} Building on URNG, we develop UG, a practical and efficient graph index that approximates the proposed framework. We design a fast interval-aware pruning technique to accelerate UG construction, and employ an iterative refinement procedure to empirically preserve monotonic searchability. In practice, the number of required iterations is small. The time complexity of each iteration is bounded by $\mathcal{O}(|C| \cdot M_{ug})$, where $|C|$ is a small constant in practice and $M_{ug}$ denotes the number of edges of UG, which ensures good scalability on large-scale datasets. We further show that UG satisfies key properties for robust performance under diverse query ranges. Based on UG, we also design an efficient query processing algorithm for graph navigation under interval-aware predicates.}

    \item {\textbf{Extensive Experiments.} We conduct extensive experiments on five datasets to evaluate the effectiveness, robustness, and scalability of UG. The results show that, for IFANN queries, UG consistently achieves the best QPS--Recall trade-off across all datasets and substantially outperforms existing hierarchical baselines in the high-recall regime. UG also remains competitive in index construction time and memory consumption, indicating that its query gains do not come at the cost of excessive indexing overhead. In addition, UG is robust under diverse filtering workloads and continues to perform strongly even for highly selective queries. Beyond IFANN, UG also achieves strong performance on ISANN and remains competitive on RFANN and RSANN despite not being specialized for these query types. Finally, scalability experiments on datasets with up to 40M vectors show that UG scales smoothly with increasing data size.}
\end{itemize}

\section{Background and Motivation}
\label{sec:preliminaries}

In this section, we first introduce the interval-aware ANN %setting
queries studied in this paper, then review the main design strategies for such queries%constrained ANN search
, and finally explain why existing methods remain insufficient for supporting diverse interval-aware workloads within a unified index.

\subsection{Data Model and Query Types}
\label{subsec:problem_def}

\noindent\textbf{Dataset.}
We consider a dataset $D$ of $n$ objects. Each object $o=(v,a_s,a_t)$ consists of a $d$-dimensional vector $o.v$ and two numerical attributes $o.a_s,o.a_t \in \mathbb{R}$ with $o.a_s \leq o.a_t$. The pair $(o.a_s,o.a_t)$ specifies the interval associated with $o$. Unless otherwise stated, vector distance is measured by the Euclidean distance.

Based on this model, we study interval-aware approximate nearest neighbor (ANN) search, where nearest neighbor retrieval is performed under query-dependent interval constraints. All queries are parameterized by a tuple $q=\langle v,I,k\rangle$, where $q.v$ is a query vector, $q.I=[a_l,a_r]$ is a query interval, and $q.k$ is the number of desired results. The query semantics differ in how object intervals are matched against $q.I$.

\noindent\textbf{Interval-Filtered ANN (IFANN) Query~\cite{yang2025hi}.}
An IFANN query returns the $q.k$ objects in $D$ with the smallest Euclidean distances to $q.v$ among those satisfying $[o.a_s,o.a_t] \subseteq q.I$.

\noindent\textbf{Interval-Stabbing ANN (ISANN) Query.}
An ISANN query returns the $q.k$ objects in $D$ with the smallest Euclidean distances to $q.v$ among those satisfying
$[o.a_s,o.a_t] \supseteq q.I$.

Two widely studied query types arise as special cases.

\noindent\textbf{Range-Filtered ANN (RFANN) Query~\cite{xu2024irangegraph,zou2026rnsgrangeawaregraphindex}.}
RFANN is a special case of IFANN in which each object is associated with a single scalar value $o.a$, i.e., $o.a_s=o.a_t=o.a$. The query returns the $q.k$ nearest objects satisfying $o.a \in q.I$.

\noindent\textbf{Range-Stabbing ANN (RSANN) Query~\cite{wang2025timestamp}.}
RSANN is a special case of ISANN in which the query interval degenerates to a single scalar value $t$, i.e., $q.I=[t,t]$. The query returns the $q.k$ nearest objects satisfying $t \in [o.a_s,o.a_t]$. This query is referred to as Timestamp ANN (TANN) in prior work~\cite{wang2025timestamp}. To unify the terminology across interval-aware ANN queries, we refer to it as RSANN in this paper, highlighting that it retrieves objects whose intervals are stabbed by the query timestamp.

Although these queries differ in interval semantics, they can be cast into a common interval-aware ANN formulation.

\subsection{Search Strategies and Graph Backbone}
\label{subsec:graph_ann}
Existing solutions to interval-aware ANN search can be understood from two complementary perspectives: how interval constraints are enforced, and what ANN backbone is used for similarity search. These two perspectives are orthogonal: interval constraints may be handled in different ways, while the underlying ANN search can still be built on a graph-based backbone.

From the perspective of constraint handling, existing methods typically fall into three categories. \emph{Pre-filtering} first identifies all objects satisfying the interval predicate and then performs ANN search over the valid subset~\cite{wang2021milvus,wei2020analyticdb}. \emph{Post-filtering} first retrieves candidates according to vector similarity and then discards those violating the predicate~\cite{li2018design}. \emph{Specialized indexing} designs dedicated structures for particular query semantics to support more efficient constrained retrieval~\cite{xu2024irangegraph,yang2025hi}.

From the perspective of ANN search, graph-based indexes have become one of the most effective backbones for modern vector retrieval. They organize objects into proximity graphs and answer queries through greedy graph traversal, achieving a strong balance between search accuracy and efficiency in high-dimensional spaces. This makes graph-based ANN a natural foundation for interval-aware retrieval.

\begin{figure}[t!] 
    \centering
    \vspace{-2em}
    \includegraphics[width=0.9\linewidth]{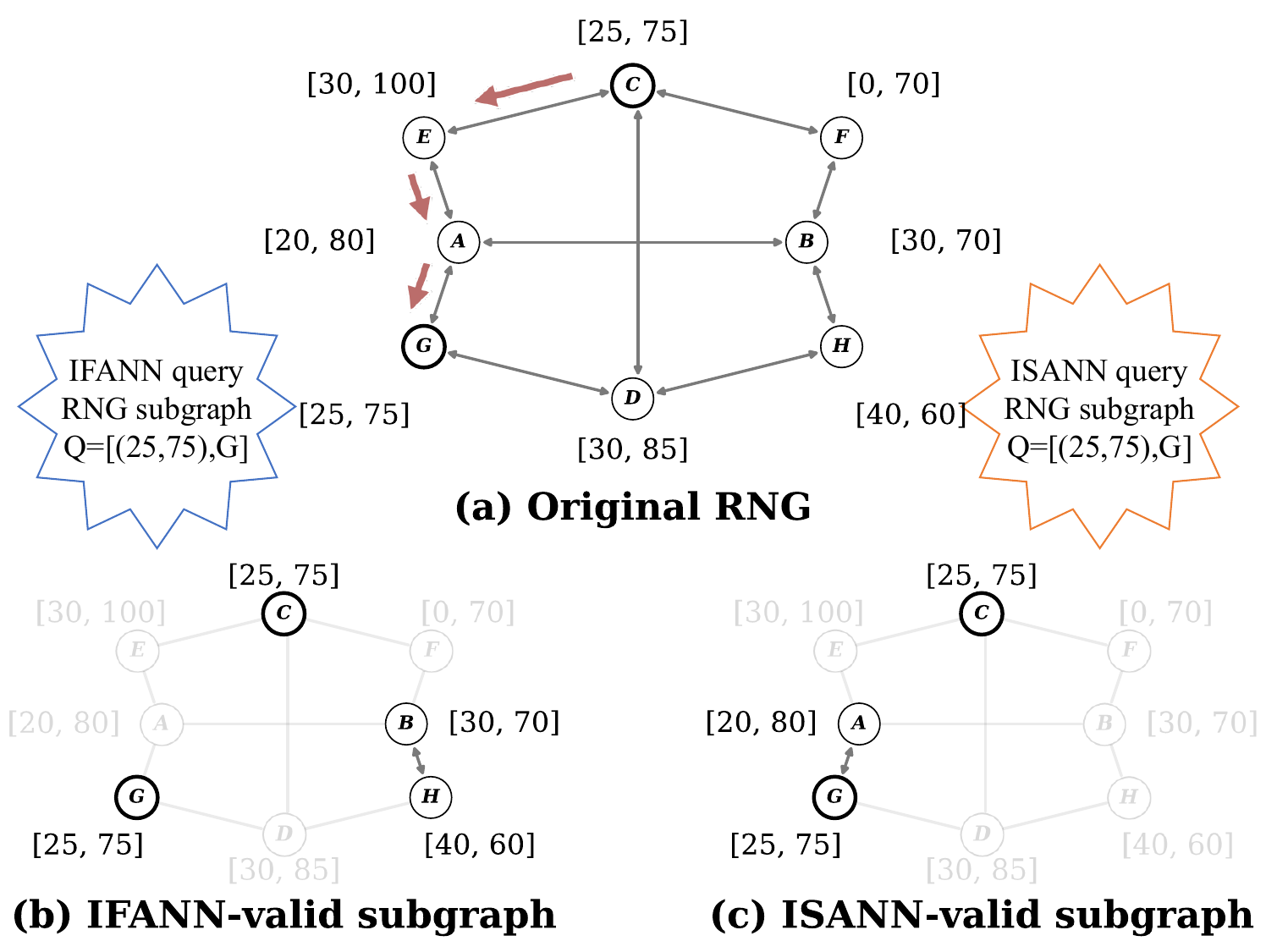} 
    \vspace{-1.5em}
    \caption{Interval constraints induce query-dependent searchable subgraphs. Starting from the same RNG, IFANN and ISANN retain different valid objects and edges, potentially invalidating critical routing structures.}
    \vspace{-2em}
    \label{fig:rng_failure}
\end{figure}

\subsection{Challenges under Interval Constraints}
\label{subsec:limitations}
Despite this common design space, existing methods remain insufficient for supporting diverse interval-aware workloads within a unified graph index. In practice, many solutions are still tailored to individual query semantics rather than designed as a shared framework across IFANN, ISANN, RFANN, and RSANN.

A natural baseline is to combine interval filtering with ANN search in a loosely coupled manner. Pre-filtering first materializes the valid subset and then performs ANN search on it, which is effective for low-selectivity queries~\cite{gollapudi2023filtered,gupta2023caps}, but can be prohibitively expensive when the valid subset remains large. Post-filtering first retrieves candidates by vector similarity and then removes invalid ones — this approach suits high-selectivity queries~\cite{wang2021milvus,wei2020analyticdb} — but often wastes substantial search effort and requires considerable oversampling before enough valid results are found. As a result, loosely coupled strategies are generally inadequate for efficient interval-aware ANN search.

More fundamentally, interval-aware constraints alter the effective search space in a query-dependent manner, posing a structural challenge to graph-based ANN search. The effectiveness of graph ANN indexes relies on good navigability in the searchable graph. Under interval-aware queries, however, search is effectively confined to the valid objects determined by the query interval, yielding a query-induced subgraph of the original graph. This subgraph may differ significantly from the full graph in both connectivity and routing structure: nodes or edges that are critical to greedy navigation in the original graph may no longer remain available after applying the interval constraint. Consequently, a graph that is highly navigable on the full dataset may become poorly connected, or even disconnected, under interval-aware queries. Figure~\ref{fig:rng_failure} illustrates this effect using an RNG-style graph\footnote{\cite{fu2017fast} refers to the graph constructed based on the RNG’s pruning strategy as the
Monotonic Relative Neighborhood Graph (MRNG). For simplicity, we use the term RNG throughout this paper.}: under the same query interval, IFANN and ISANN induce different valid subgraphs, and routing structures available in the original graph may disappear after filtering.

Some specialized graph indexes partially alleviate this issue for their target query semantics. For example, strong IFANN solutions often rely on partitioned or multi-structure designs, such as hierarchical interval partitioning~\cite{yang2025hi}, to achieve high query efficiency. However, these approaches incur additional indexing and memory overhead, complicate updates and maintenance, and do not naturally generalize to other interval-aware query types. More broadly, existing specialized methods remain tightly coupled to specific predicates and therefore do not generalize naturally across interval-aware query types.

These observations suggest that the key challenge is not merely to optimize search for one particular interval predicate, but to design a unified graph framework whose search effectiveness is preserved under diverse query-dependent filtering conditions. This challenge motivates the framework developed in the next section.

\comment{
\subsection{Limitations of Existing Strategies}
\label{subsec:limitations}

Despite the success of graph-based ANNS in pure vector spaces, extending them to handle IF-ANNS queries remains a significant challenge. Existing solutions generally fall into three categories, each suffering from inherent structural limitations.

\subsubsection{Loose Coupling Strategies (Pre/Post-filtering)}
The most straightforward approach is to decouple the vector search from the interval filtering.
\begin{itemize}
    \item \textbf{Pre-filtering:} This method first retrieves all objects satisfying the interval constraint $I_o \subseteq I_q$ using an inverted index or a tree structure (e.g., Segment Tree), and then performs a linear scan or vector search on the result set. However, when the query selectivity is low (i.e., many objects satisfy the interval), the candidate set becomes excessively large, degrading the performance to that of a brute-force scan.
    \item \textbf{Post-filtering:} This method performs a standard vector search on the global graph to retrieve top-$k'$ candidates and then filters out those failing the interval constraint. As noted in recent studies, this approach suffers from the "oversampling" problem. To obtain $k$ valid results, the searcher must explore $k' \gg k$ candidates. In the worst case (e.g., disjoint intervals), the search may fail to return enough valid results even after exhaustive exploration.
\end{itemize}

\subsubsection{Partition-based Indexing}
To address the inefficiency of loose coupling, partition-based methods like \textbf{Hi-PNG} and others employ a "divide-and-conquer" strategy.
Hi-PNG constructs a hierarchical quad-tree to partition the interval space into disjoint "fields" and builds separate proximity graphs for each leaf node[cite: 196, 197]. During a query, the system identifies relevant partitions and searches within their local graphs.

While effective for specific distributions, this approach introduces \textbf{structural fragmentation}. The geometric neighborhood of a data point is artificially sliced by partition boundaries. A query spanning multiple partitions forces the search algorithm to traverse multiple disjoint graphs or backtrack through the tree structure[cite: 223], breaking the desirable "single-graph, single-traversal" property. Furthermore, maintaining a dual index (tree + graphs) increases memory overhead and implementation complexity.

\subsubsection{Naive In-filtering and Connectivity Rupture}
The third strategy, \textit{In-filtering}, attempts to perform filtering on-the-fly during the graph traversal. The search algorithm navigates the global graph but discards invalid nodes ($I_o \not\subseteq I_q$) from the candidate queue.
While this avoids the overhead of partitioning, applying it to a standard proximity graph (e.g., HNSW, NSG) leads to the critical issue of \textbf{Connectivity Rupture}.

Standard graphs are constructed based solely on vector proximity. A node $u$ relies on specific "bridge" neighbors to route the search to distant targets. In the context of IF-ANNS, these bridge nodes may be geometrically optimal but fail the interval constraint for a specific query $I_q$.
When the filter implicitly removes these critical bridges, the induced subgraph of valid nodes becomes fragmented into disconnected components. Consequently, the greedy beam search gets trapped in local optima, unable to reach the true nearest neighbors located in other components.
This phenomenon highlights the fundamental need for a graph structure that is not only globally connected but also \textit{resiliently connected} under arbitrary interval filtering—a property we define as \textit{Structural Heredity} in the next section.
}

% ---------------------------------------------------------
% 3. Unified Interval-aware Relative Neighborhood Graph
% ---------------------------------------------------------

\section{Unified Interval-aware RNG}
\label{sec:urng}

In this section, we establish the theoretical foundation of the \emph{Unified Interval-aware Relative Neighborhood Graph} (URNG).
We proceed to formally define the URNG structure and provide a rigorous analysis of its graph-theoretic properties, including monotonic searchability, structural heredity, and asymptotic complexity bounds.

\begin{definition}[URNG]
\label{def:urng}
Let \( D \) be a dataset in a metric space with distance function \( \delta \), where each object is associated with an interval attribute \( I \).
The \emph{Unified Interval-aware Relative Neighborhood Graph} (URNG) on \( D \) is defined as a graph $G = (V, E, \mathrm{st})$,
where each node in \( V \) corresponds to an object in \( D \), and each edge \( (u,v) \in E \) is associated with a semantic bitmask $\mathrm{st}(u,v) = \bigl(b_{\mathrm{IF}}(u,v),\, b_{\mathrm{IS}}(u,v)\bigr) \in \{0,1\}^2.$ For any two distinct nodes \( u,v \in V \), the \textbf{IF bit} \( b_{\mathrm{IF}}(u,v) \) is set to \(1\) if and only if there does \emph{not} exist a witness node \( w \in V \setminus \{u,v\} \) such that all of the following conditions hold\footnote{Here, the union symbol for two intervals, $I_A \cup I_B$, denotes the interval $\min(I_A.l, I_B.l), \max(I_A.r, I_B.r)$. For convenience, throughout the subsequent description of the pruning conditions, $\cup$ is used with this meaning.}:
\[
\delta(u,w) < \delta(u,v), \delta(v,w) < \delta(u,v), I_w \subseteq I_u \cup I_v, \textit{and}\ b_{\mathrm{IF}}(u,w) = 1\]

Similarly, the \textbf{IS bit} \( b_{\mathrm{IS}}(u,v) \) is set to \(1\) if and only if there does \emph{not} exist a witness node \( w \in V \setminus \{u,v\} \) such that all of the following conditions hold:
\[
\delta(u,w) < \delta(u,v),\delta(v,w) < \delta(u,v),I_u \cap I_v \subseteq I_w ,\textit{and}\ b_{\mathrm{IS}}(u,w) = 1
\]

Accordingly, an edge \( (u,v) \) is active for \emph{IFANN queries} when \( b_{\mathrm{IF}}(u,v)=1 \), active for \emph{ISANN queries} when \( b_{\mathrm{IS}}(u,v)=1 \), and active for both query types when $\mathrm{st}(u,v) = (1,1)$.

\end{definition}

%This unified formulation will also facilitate the transition to the practical index introduced in Section~4, where the same bitmask-based representation is adopted for efficient construction and query processing.

\begin{figure}[t!] % 星号表示跨双栏，[t!] 强行置于页顶
    \centering
    % 请根据你的实际文件夹结构替换这里的路径
    \vspace*{-2em}
    \includegraphics[width=0.9\linewidth]{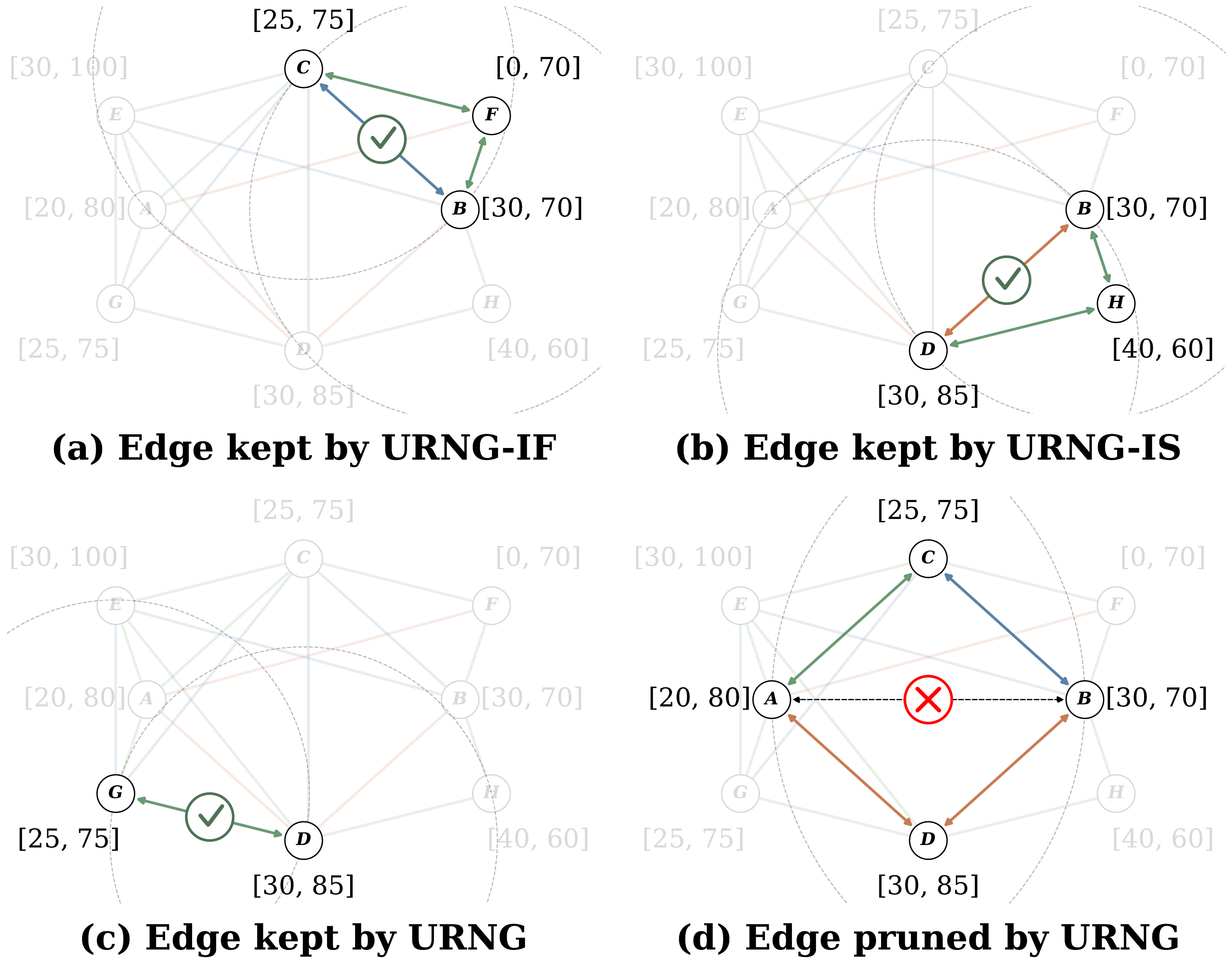} 
    \vspace*{-1em} % 如果觉得图和标题之间太空，可以微调这个值
    \caption{Illustration of the URNG pruning strategy}
    \vspace*{-1.5em}
    \label{fig:urng_pruning}
\end{figure}

% the above definition is expressed through the semantic status of witness edges, it is well-defined. In particular, whether an edge $(u,v)$ remains active depends only on witness edges $(u,w)$ that are strictly shorter than $(u,v)$. Therefore, the semantic state of URNG edges is uniquely determined, and the above definition can equivalently be understood as an inductive definition in ascending order of edge length.

\paragraph{Example 1.}
% TODO: Add the formal example description here.

Figure~\ref{fig:urng_pruning} presents four representative cases that illustrate how edges are retained or pruned in URNG under different interval semantics.

\begin{enumerate}[leftmargin=*,nosep]
    \item {The edge \((B,C)\) is the longest edge in triangle \((B,C,F)\), and would therefore be pruned in a classical RNG. Under interval-aware semantics, however, only the IS constraint $[30,70] \subseteq [0,70]$
    is satisfied. As a result, the IS functionality of \((B,C)\) is pruned, while its IF functionality is preserved.}

    \item {The edge \((B,D)\) is the longest edge in triangle \((B,D,H)\), and would also be removed in a classical RNG. Under interval-aware semantics, only the IF constraint
    $[40,60] \subseteq [30,85]$
    is satisfied. Hence, the IF functionality of \((B,D)\) is pruned, while its IS functionality remains active.}

    \item {For the edge \((D,G)\), there exists no triangle in which it becomes the longest edge. Therefore, this edge cannot be pruned by any valid witness node and remains active under both semantics.}

    \item {The edge \((A,B)\) is the longest edge in triangle \((A,B,C)\), and since
    $
    [25,75] \subseteq [20,80],
    $
    its IF functionality is pruned by edge \((A,C)\). The same edge \((A,B)\) is also the longest edge in triangle \((A,B,D)\), and since
    $
    [30,70] \subseteq [30,85],
    $
    its IS functionality is pruned by edge \((A,D)\). Since both semantic bits are removed, the edge is completely discarded from the graph.}
\end{enumerate}

The above definition unifies the two interval semantics into a single theoretical graph. Compared with defining two independent proximity graphs, URNG %associates each edge with a semantic status code, so that the IF and IS subgraphs can be viewed as two semantic projections of the same underlying graph. This 
offers two advantages: (1) reduced memory cost, as nodes and many edges are shared across the two graphs; and (2) efficient construction and query processing due to the bitmask-based representation, as detailed in Section~\ref{sec:ug}.

\noindent\textbf{Difference from the classical RNG.} 
% As suggested by Figures~\ref{fig:rng_failure} and~\ref{fig:urng_pruning}, 
There is no direct inclusion relationship between RNG and URNG. 
The reason is fundamental: in URNG, pruning depends not only on the geometric condition that an edge is the longest edge of a triangle, but also on whether the witness node remains valid under the corresponding interval semantics.
Consequently, some edges that would be removed in a classical RNG may be retained in URNG due to the lack of a semantically valid witness.
Conversely, such retained edges may later act as new semantic witnesses and prune other edges that would otherwise survive in the classical RNG, as illustrated in Figure ~\ref{fig:rng_failure} and~\ref{fig:urng_pruning}.

However, it is precisely URNG's awareness of interval that allows it to maintain connectivity and high performance across different queries. As illustrated in Figure~\ref{fig:urng_path}, navigation from \(A\) to \(B\) in URNG does not necessarily rely on the direct edge \((A,B)\).
In the IF-induced subgraph, the path \((A,C),(C,B)\) shown in Figure~(a) can replace the direct edge.
In the IS-induced subgraph, the path \((A,D),(D,B)\) shown in Figure~(b) plays an analogous role.
This witness-mediated routing mechanism is precisely what prevents the filtered subgraph from becoming disconnected under interval constraints, and it forms the structural basis for the advantage of URNG in interval-aware search.

\subsection{Theoretical Properties}
\label{subsec:urng_properties}

In this subsection, we introduce two key properties of URNG that enable it to support various interval-aware ANN queries within a single graph structure.

\begin{definition}[Monotonic path]
\label{Mp}
Given a graph \( G=(V,E) \) and a distance metric \( \delta \), a path
$=v_0 \to v_1 \to v_2 \to \cdots \to v_k=t
$
is called a \emph{monotonic path} with respect to the target node \( t \) if $
\delta(v_{i+1}, t) < \delta(v_i, t), \forall\, 0 \le i < k.$
That is, the distance to the target strictly decreases at every step along the path.
\end{definition}

% We now prove the existence of a monotonic path.
We now prove the existence of a monotonic path in each semantic projection of URNG.

% \begin{theorem}[Monotonic Searchability]
% \label{thm:monotonicity}
% \yin{For any dataset ${D}$ and the URNG $G$ constructed over $D$, consider any two nodes $s, t \in V$. A path $s = v_0 \to \cdots \to v_m = t$ along which the distance to $t$ strictly decreases at each step, i.e., $\delta(v_{i+1}, t) < \delta(v_i, t)$ for all $0 \le i < m$, is called a monotonic path. Notably, in a URNG, a monotonic path exists between every pair of nodes.}
% \end{theorem}

\begin{theorem}[Monotonic Searchability]
\label{thm:monotonicity}
{For any dataset $D$ and a 
URNG graph $G$ constructed over $D$, fix one semantic projection $G^\sigma$,
where $\sigma \in \{\mathrm{IF}, \mathrm{IS}\}$. For any two nodes $s, t \in V$, a
monotonic path from $s$ to $t$ exists in $G^\sigma$.}
\end{theorem}

% \begin{proof}
% \yin{Without loss of generality, consider the current node $u \neq t$. If $(u, t) \in E$, the condition is satisfied directly. Otherwise, if $(u, t) \notin E$, by Definition~\ref{def:urng}, there exists a sentinel node $w$ such that $(u, w) \in E$ and $\delta(w, t) < \delta(u, t)$. The search can thus safely proceed to $w$. By induction, the distance to $t$ strictly decreases at each step until $t$ is reached.}
% \end{proof}

\begin{proof}
{Without loss of generality, consider the current node $u \ne t$.
If $(u,t) \in E^\sigma$, the condition is satisfied directly. Otherwise,
if $(u,t) \notin E^\sigma$, then by Definition 3.1 under the fixed
semantics $\sigma$, there exists a witness node $w$ such that
$(u,w) \in E^\sigma$ and $\delta(w,t) < \delta(u,t)$. The search can thus
safely proceed from $u$ to $w$. Repeating the same argument, the
distance to $t$ strictly decreases at every step. Since the dataset is
finite, this process must terminate, and the only possible terminal
node is $t$. Therefore, a monotonic path from $s$ to $t$ exists in
$G^\sigma$.}
\end{proof}

\begin{figure}[t!] % 星号表示跨双栏，[t!] 强行置于页顶
    \centering
    % 请根据你的实际文件夹结构替换这里的路径
    \vspace*{-2em}
    \includegraphics[width=0.9\linewidth]{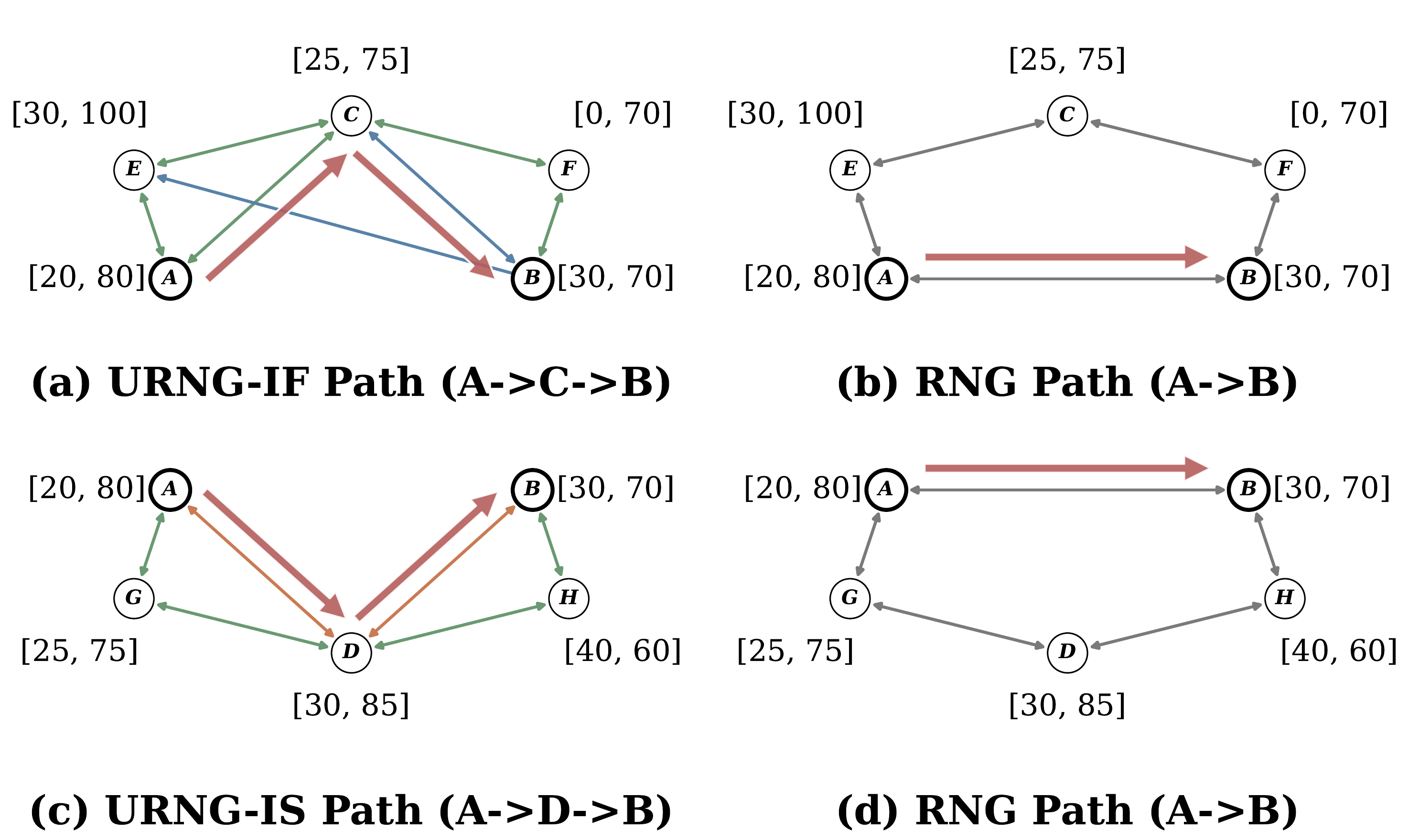} 
    
    \vspace*{-1em} % 如果觉得图和标题之间太空，可以微调这个值
    \caption{Differences between RNG and URNG.}
    \vspace*{-2em}
    \label{fig:urng_path}
\end{figure}

\paragraph{Example 2.}
% TODO: Add the formal example description here.

As shown in Figure~\ref{fig:urng_path}, under IF filtering, the path $A \to C \to B$ exists in Figure (a), while under IS filtering, the path $A \to D \to B$ exists in Figure (c). Figure (b) and (d) further illustrate the direct edge $A \to B$.These examples show that although URNG reorganizes the local graph structure according to interval semantics, it still preserves the core navigability property required for search. 

The monotonicity property provides the theoretical justification for applying graph traversal over any semantic projection of URNG. It is worth noting that classical RNG also belongs to the family of monotonic search networks and thus admits monotonic paths as well ~\cite{fu2017fast}. %As shown in Figure 3, under IF filtering, the path $A \to C \to B$ exists in Figure (a), while under IS filtering, the path $A \to D \to B$ exists in Figure (c). Figures (b) and (d) further illustrate the direct edge $A \to B$.These examples show that although URNG reorganizes the local graph structure according to interval semantics, it still preserves the core navigability property required for search. 

\begin{corollary}
\label{cor:beam_correctness}
Let $G^\sigma = (V, E^\sigma)$ be a fixed semantic projection of the URNG constructed on dataset $D$ under distance metric $\delta$,where $\sigma \in \{\mathrm{IF}, \mathrm{IS}\}$. For any target node $t \in V$, beam search using $t.v$ as the query vector can always reach $t$ regardless of the choice of the entry node.
\end{corollary}

\begin{proof}
Without loss of generality, let $s$ be the entry node. By
Theorem 3.3, there exists a monotonic path
$s = v_0 \to v_1 \to \cdots \to v_m = t$ in $G^\sigma$ such that
$\delta(v_{i+1}, t) < \delta(v_i, t)$ for all $0 \le i < m$. Since the query
vector is exactly $t.v$, we have
$\delta(v_{i+1}, q.v) < \delta(v_i, q.v)$ at every step along the same path.
Therefore, the search process always has access to at least one candidate
that is closer to the query target than the current node. Since beam
search always keeps the closest candidates and the dataset is finite, the
search process eventually reaches $t$.
\end{proof}

% In addition to monotonic searchability, URNG satisfies another crucial property, namely \emph{structural heredity}.
% This property states that filtering a globally constructed URNG and directly rebuilding URNG on the filtered node set produce equivalent graph structures.
% This is precisely the property that classical RNG lacks in interval-aware settings. In the following theorem, we fix one query semantics at a time. Accordingly, $V_I$ denotes the query-valid node set under the chosen semantics, and $G[I]$ denotes the induced subgraph relevant to that semantics. That is, for IFANN queries, $V_I$ consists of nodes whose intervals are contained in $I$, while for ISANN queries, $V_I$ consists of nodes whose intervals contain $I$. We have the following theorem.

In addition to monotonic searchability, URNG satisfies another crucial property, namely \emph{structural heredity}. This property states that filtering a globally constructed URNG and directly rebuilding URNG on the filtered node set produce equivalent graph structures. This is precisely the property that classical RNG lacks in interval-aware settings. In the following theorem, we fix one query semantics at a time. Accordingly, $V_I$ denotes the query-valid node set under the chosen semantics, and $G[I]$ denotes the induced subgraph relevant to that semantics. We then have the following theorem.

\begin{theorem}[Structural Heredity]
\label{thm:heredity}
Let \(G=(V,E)\) be the URNG constructed on dataset \(D\) under distance metric \(\delta\).
For any query
$
q=\langle v, I=[a_l,a_r], k \rangle,
$
let \(G[I]=(V_I,E_I)\) denote the induced subgraph of \(G\), and let \(G'=(V',E')\) denote the URNG directly constructed on the node set \(V_I\).
Then
$
G[I] \equiv G',
$
that is, the two graphs are structurally equivalent.
\end{theorem}

\begin{proof}
Since \(V_I = V'\), it suffices to prove \(E_I = E'\). We use induction on edge length.

For the base case, let \((u,v)\) be the shortest candidate edge with both endpoints in \(V_I\). Suppose \((u,v)\) were pruned by some witness node \(w \in V\). Then \((u,w)\) must be shorter than \((u,v)\). Moreover, the semantic condition forces \(w \in V_I\): for IFANN, \([w.l,w.r] \subseteq I\); for ISANN, \(I \subseteq [w.l,w.r]\). Thus \((u,w)\) is a shorter edge inside \(V_I\), contradicting the choice of \((u,v)\). Hence \((u,v)\) cannot be pruned, and its semantic state is identical in \(E_I\) and \(E'\).

Now assume that all edges shorter than \((u,v)\) have the same semantic states in \(E_I\) and \(E'\). If a semantic bit of \((u,v)\) remains active in \(E_I\), then no valid witness with the same semantic status exists in \(V_I\); since \(E'\) is constructed on the same node set \(V_I\), no such witness can appear in \(E'\), so the bit also remains active there.

Conversely, if a semantic bit of \((u,v)\) is inactive in \(E_I\), then there exists a witness edge \((u,w)\) satisfying the corresponding pruning condition. By the same argument as above, the semantic condition implies \(w \in V_I\). Since \(\mathrm{dis}(u,w) < \mathrm{dis}(u,v)\), the induction hypothesis guarantees that \((u,w)\) has the same semantic state in \(E'\), and therefore \((u,v)\) is also pruned in \(E'\).

Thus \((u,v)\) has the same semantic state in both graphs. By induction, all edges do, and hence \(E_I = E'\).
\end{proof}

The heredity property guarantees that, %has an immediate algorithmic implication.
for any interval-aware ANN query, it is unnecessary to explicitly reconstruct the graph structure from scratch on the filtered node set. 
Instead, one can directly search over the corresponding induced subgraph, and the result is equivalent to that of rebuilding URNG after filtering. 
This is exactly the structural advantage that URNG provides over the classical RNG in interval-aware search.

\begin{figure}[t!] % 星号表示跨双栏，[t!] 强行置于页顶
    \centering
    % 请根据你的实际文件夹结构替换这里的路径
    \vspace*{-2em}    
    \includegraphics[width=0.9\linewidth]{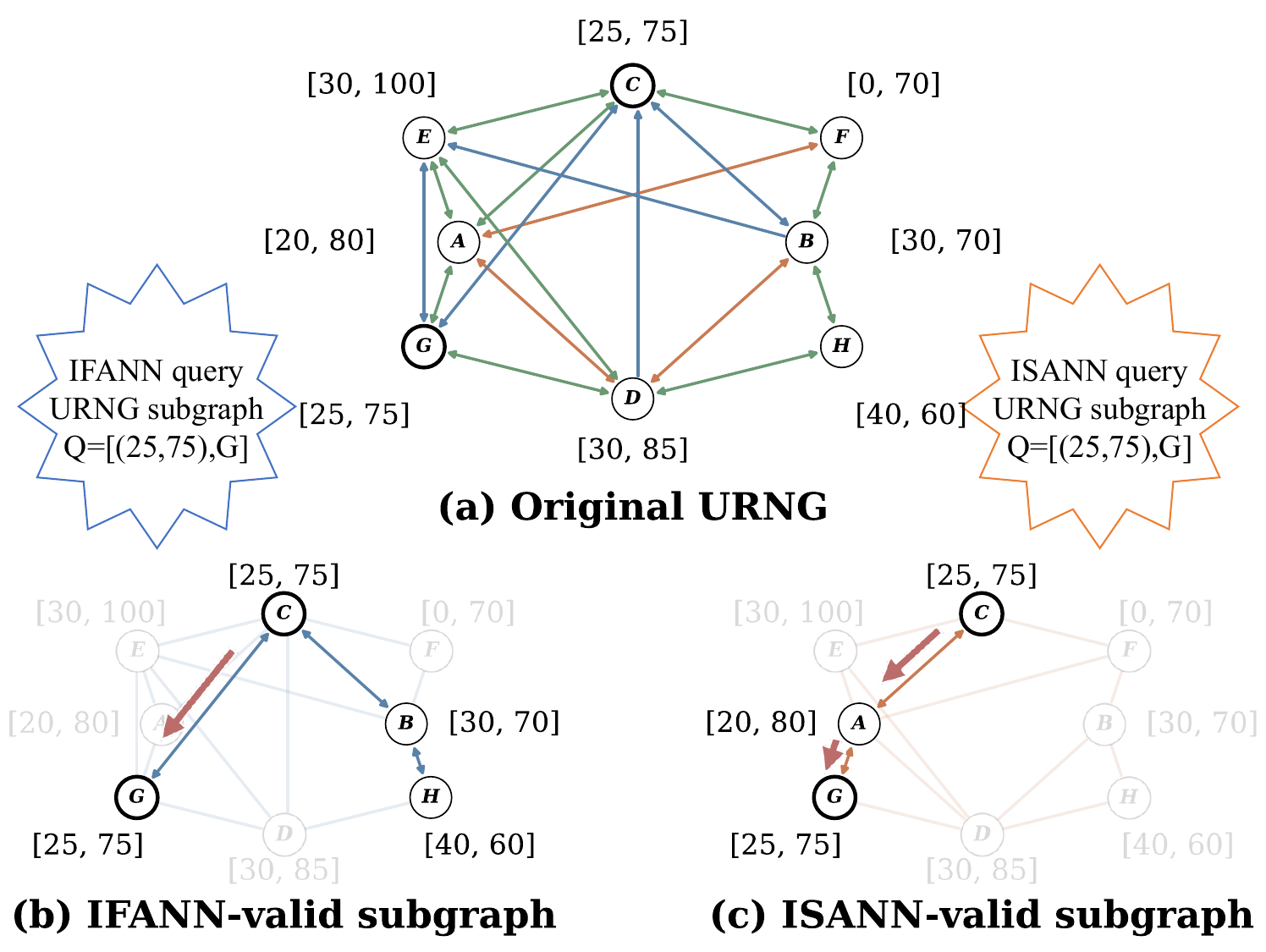}     
    \vspace*{-1em} % 如果觉得图和标题之间太空，可以微调这个值
    \caption{Illustration of the hereditary property of URNG.}
    \vspace*{-2em}
    \label{fig:urng_connected}
\end{figure}

Figure~\ref{fig:urng_connected} illustrates the induced URNG subgraphs for the example shown in Figure~\ref{fig:rng_failure}.
Under the interval filter \([25,75]\), both the IF-induced and IS-induced subgraphs remain well connected and navigable, thereby supporting the intended search process and returning the desired results.

\subsection{Complexity Analysis}
\label{subsec:urng_complexity}

We next summarize the complexity of URNG under a uniform interval model, where interval endpoints are generated independently and uniformly. Detailed proofs are deferred to the appendix~\cite{ug_fullversion}.
%Appendix ~\ref{app:complexity_proofs}.

\begin{theorem}
\label{thm:urng_search_complexity}
Under the assumption that interval attributes are independently and uniformly distributed, the expected complexity of interval-aware beam search on URNG is \(O(n^{1/d}\log n)\). 
\end{theorem}

\begin{theorem}
\label{thm:urng_index_size}
The index size of URNG is upper bounded by \(O(C_{\mathrm{urng}}S_r)\\=O(S_r)\), where \(S_r\) denotes the index size of the corresponding RNG and \(C_{\mathrm{urng}}\) is a constant independent of \(n\).
\end{theorem}

These results show that, under the uniform interval model, URNG preserves the navigability of monotonic search graphs while incurring only a constant-factor overhead over the corresponding RNG.

\begin{theorem}
\label{thm:exact_build}
The time complexity of exact URNG construction is \(O(n^3)\).
\end{theorem}

This $O(n^3)$ exact-construction cost further motivates the practical UG index introduced in Section~4.
% 

% \subsection{Relationship to RRNG}
\noindent\textbf{Difference from the (RRNG)~\cite{zou2026rnsgrangeawaregraphindex}.} The proposed URNG is closely related to the Range-aware Relative Neighborhood Graph (RRNG)~\cite{zou2026rnsgrangeawaregraphindex}, which was originally developed for RFANN search, since both aim to preserve monotonic searchability and structural heredity under attribute-induced filtering. However, URNG substantially generalizes RRNG in both the data model and query semantics: while RRNG assumes scalar attributes and RFANN  queries, URNG handles interval attributes and supports multiple interval-inclusion semantics, including ISANN, IFANN, RSANN and RFANN search. Indeed, RRNG can be viewed as a special case of URNG when each object interval degenerates to a point, i.e., $I_o=[o.a,o.a]$, and only the IF projection is considered.

\section{The Unified Interval-aware Graph-based Index}
\label{sec:ug}

\begin{algorithm}[t]
\caption{UG Initial Candidate Generation}
\label{alg:ug_cand}
\SetKwInput{Input}{Input}
\SetKwInput{Output}{Output}
\SetKwProg{Fn}{Function}{}{end}

\Input{Dataset $V$, spatial budget $ef_{spatial}$, attribute budget $ef_{attribute}$}
\Output{Initial candidate sets $\mathcal{C}$ for all nodes}

\Fn{\textsc{GenerateCandidates}($V, ef_{spatial}, ef_{attribute}$)}{
    $\mathcal{C}_{spa} \leftarrow \textsc{NNDescent}(V, ef_{spatial})$\;
    
    \ForEach{$u \in V$}{
        $\mathcal{C}_{attr}(u) \leftarrow \emptyset$\;
    }
    
    $Keys \leftarrow \{l, r, mid, len\}$\;
    \ForEach{$k \in Keys$}{
        $V_k \leftarrow$ sort $V$ in ascending order of attribute $k$\;
        \ForEach{$u \in V_k$}{
            $\mathcal{N}_k(u) \leftarrow$ fetch $\frac{ef_{attribute}}{8}$ adjacent nodes from both sides of $u$ in $V_k$\;
            $\mathcal{C}_{attr}(u) \leftarrow \mathcal{C}_{attr}(u) \cup \mathcal{N}_k(u)$\;
        }
    }
    
    \ForEach{$u \in V$}{
        $\mathcal{C}(u) \leftarrow \mathcal{C}_{spa}(u) \cup \mathcal{C}_{attr}(u)$\;
        $\mathcal{C}(u) \leftarrow \text{Unique}(\mathcal{C}(u)) \setminus \{u\}$\;
    }
    
    \KwRet $\mathcal{C}$\;
}
\end{algorithm}

In Section~3, we introduced the unified theory of the \emph{Unified Interval-aware Relative Neighborhood Graph} (URNG). %The central insight is that if the witness node of an edge remains valid under the query semantics, then the resulting graph simultaneously preserves monotonic searchability and the heredity of query-induced subgraphs. These properties make URNG a principled graph structure for interval-constrained approximate nearest neighbor search. 
However, constructing an exact URNG requires explicitly evaluating witness relationships on the complete graph, where % each of the $O(n^2)$ candidate edges may need to be checked against $O(n)$ possible witness nodes. This 
results in an $O(n^3)$ %worst-case 
construction cost %, which is prohibitive for large-scale datasets 
and %therefore 
prevents %exact
URNG from being %directly 
applied in practice. To address this issue, we propose the \emph{Unified Interval-aware Graph-based index} (UG), a practical graph index that approximates URNG effectively while preserving its structural advantages.  %as much as possible. 

In particular, UG is built in two stages. We first generate a high-quality candidate set for each node, and then apply an iterative unified pruning strategy to derive the final graph. In addition, we design an interval-aware beam search algorithm over UG to support efficient navigation for both \emph{IFANN queries} and \emph{ISANN queries}. %A key design choice is that UG does not maintain two separate physical graphs for the two query types. Instead, it stores all edges in a single graph and associates each edge with a semantic bitmask indicating under which query semantics the edge is active.

% \vspace*{-1em}

\subsection{Candidate Generation}
\label{subsec:ug_candidate}

The computational bottleneck of exact URNG construction lies in the size of the candidate set. In principle, exact pruning would require considering all nodes as potential neighbors. However, the analysis in Section~3 suggests that the edges that ultimately survive pruning are highly selective: an edge is retained not merely because of geometric proximity, but rather due to the absence of any semantically valid witness that can eliminate it under the query semantics. This observation motivates us to replace the complete candidate set with a high-quality approximate candidate pool.

Algorithm~\ref{alg:ug_cand} constructs the initial candidate pool for each node by combining spatial proximity and interval proximity. 
It first applies \textsc{NNDescent} to the whole dataset with budget $ef_{spatial}$, obtaining a spatial candidate set $\mathcal{C}_{spa}$ that captures the local neighborhood structure in the vector space (line~2). To further introduce interval-relevant candidates, the algorithm initializes an empty attribute-aware candidate set for each node (lines~3--4), and then examines four interval-derived keys, including the left endpoint $l$, the right endpoint $r$, the midpoint $\mathit{mid}$, and the interval length $\mathit{len}$ (line~5). For each key, all nodes are sorted according to the corresponding interval attribute (line~7). 
The algorithm then scans the sorted order and collects nearby nodes around each node $u$ as attribute-aware candidates (lines~8--10). 
In particular, it fetches up to $\lfloor ef_{attribute}/8 \rfloor$ adjacent nodes from each side of $u$ under the current ordering, so that the overall number of interval-based candidates is controlled by $ef_{attribute}$ across the four keys and two directions (line~9). 
These candidates are useful because nodes close in such interval-derived orders are more likely to survive the same interval constraint or serve as semantic witnesses in later pruning. 
Finally, UG merges the spatial and attribute-aware candidates, removes duplicates and the node itself, and returns the initial candidate sets for all nodes (lines~11--14).

Candidate generation in UG is motivated by the observation in Section~\ref{sec:urng}
that URNG edges are determined by both spatial proximity and semantically
valid witnesses. Accordingly, UG constructs the candidate pool from two
complementary sources: spatial neighbors, which provide the basic
navigational backbone, and interval-aware neighbors, which increase the
chance of covering useful IF/IS witnesses. Since exact URNG construction
would require considering all nodes, this stage only provides a bounded
candidate space rather than enforcing URNG properties by itself. The
URNG-like structure is further approximated in the subsequent stages:
unified pruning applies the interval-aware witness conditions within the
candidate pool, while iterative repair compensates for missing continuation
paths caused by the bounded candidate set. Empirically, this design provides
sufficient candidate coverage for high-recall search, as shown in Section~\ref{sec:experiments}.

\begin{algorithm}[t]
\caption{UG Iterative Construction}
\label{alg:ug_build}
\SetKwInput{Input}{Input}
\SetKwInput{Output}{Output}
\SetKwProg{Fn}{Function}{}{end}

\Input{Dataset $V$, initial candidates $\mathcal{C}^{(0)}$, semantic degree budget $M$, iterations $T$}
\Output{Neighbor sets $\mathcal{N}^{(T)}$ with semantic bitmasks}

\Fn{\textsc{BuildUG}($V, \mathcal{C}^{(0)}, M, T$)}{
    \ForEach{$u \in V$}{
        $\mathcal{W}^{(0)}(u) \leftarrow \emptyset$\;
    }
    
    \For{$t \leftarrow 1$ \KwTo $T$}{
        \ForEach{$u \in V$}{
            $\widehat{\mathcal{C}}^{(t)}(u) \leftarrow \textsc{Unique}(\mathcal{C}^{(t-1)}(u) \cup \mathcal{W}^{(t-1)}(u))$\;
            $\mathcal{W}^{(t)}(u) \leftarrow \emptyset$\;
        }
        
        \ForEach{$u \in V$}{
            $\mathcal{N}^{(t)}(u), \Delta \mathcal{W} \leftarrow \textsc{UnifiedPrune}(u, \widehat{\mathcal{C}}^{(t)}(u), M)$\;
            $\mathcal{C}^{(t)}(u) \leftarrow \{\, v \mid (v,\cdot) \in \mathcal{N}^{(t)}(u) \,\}$\;
            
            \ForEach{$(w,v) \in \Delta \mathcal{W}$}{
                $\mathcal{W}^{(t)}(w) \leftarrow \mathcal{W}^{(t)}(w) \cup \{v\}$\;
            }
        }
    }
    
    \KwRet $\mathcal{N}^{(T)}$\;
}
\end{algorithm}

\subsection{Pruning Strategy%UG Construction via Iterative Unified Pruning
}
\label{subsec:ug_build}

After obtaining the candidate graph, UG constructs the final index by pruning redundant edges while preserving, as much as possible, the structural behavior of URNG. The goal of this stage is to approximate the unified pruning semantics of URNG within a bounded candidate pool, while keeping the graph sparse enough for efficient search. %Although the maximum degree is controlled by a parameter \(M\), this bound mainly serves as an engineering constraint; the core design lies in how pruning is performed under the two interval semantics.
To support both IFANN queries and ISANN queries within a single graph, UG associates each retained edge with a semantic bitmask indicating whether the edge is active under IF semantics, IS semantics, or both. Therefore, UG does not explicitly maintain two separate graphs. Instead, it stores a unified edge set, and later query processing dynamically selects the %appropriate logical 
subgraph according to the query type.

Algorithm~\ref{alg:ug_build} presents the iterative construction procedure of UG. The algorithm starts from the initial candidate sets $\mathcal{C}^{(0)}$ and maintains a repair candidate set $\mathcal{W}^{(t)}$ for each iteration. Initially, the repair set of every node is empty (lines~2--3). At the beginning of the $t$-th iteration, UG forms a refined candidate pool for each node by merging the candidates retained from the previous round and the repair candidates generated in the previous round (lines~5--7). This design avoids accumulating all historical candidates, while still allowing potentially useful pruned endpoints to be reconsidered. After the refined candidate pool is obtained, UG applies \textsc{UnifiedPrune} to each node, producing a semantic neighbor set $\mathcal{N}^{(t)}(u)$ and a set of repair pairs $\Delta\mathcal{W}$ (lines~8--9). The retained neighbors are extracted as the base candidate set for the next iteration (line~10). For each repair pair $(w,v)$, the pruned endpoint $v$ is inserted into the repair set of the witness node $w$, so that the possible continuation path through $w$ can be explored in the next round (lines~11--12). After $T$ iterations, the algorithm returns the final neighbor sets $\mathcal{N}^{(T)}$, where each retained edge is associated with its semantic bitmask (line~13).

\begin{algorithm}[t]
\caption{UG Unified Prune}
\label{alg:unified_prune}
\SetKwInput{Input}{Input}
\SetKwInput{Output}{Output}
\SetKwProg{Fn}{Function}{}{end}

\Input{Node $u$, candidate set $\mathcal{C}(u)$, semantic degree budget $M$}
\Output{Pruned neighbors $\mathcal{N}(u)$ and repair set $\Delta\mathcal{W}$}

\Fn{\textsc{UnifiedPrune}($u,\mathcal{C}(u),M$)}{
    $\mathcal{N}(u)\leftarrow\emptyset,\ \Delta\mathcal{W}\leftarrow\emptyset$\;
    $cnt_{IF}\leftarrow0,\ cnt_{IS}\leftarrow0$\;
    sort $\mathcal{C}(u)$ by increasing $\delta(u,\cdot)$\;
    
    \ForEach{$v\in\mathcal{C}(u)$}{
        $s_v[IF]\leftarrow1,\ s_v[IS]\leftarrow1$\;
        \If{$I_u\cap I_v=\emptyset$}{
            $s_v[IS]\leftarrow0$\;
        }
        
        \ForEach{$w\in\mathcal{N}(u)$}{
            \If{$s_v=0$}{\textbf{break}\;}
            \If{$\delta(v,w)\ge\delta(u,v)$}{\textbf{continue}\;}
            
            \ForEach{$\sigma\in\{IF,IS\}$}{
                \If{$s_v[\sigma]=1$ \textbf{and} $w.status[\sigma]=1$ \textbf{and} $\Phi_\sigma(u,v,w)$}{
                    $s_v[\sigma]\leftarrow0$\;
                    $\Delta\mathcal{W}\leftarrow\Delta\mathcal{W}\cup\{(w,v)\}$\;
                }
            }
        }
        
        \ForEach{$\sigma\in\{IF,IS\}$}{
            \If{$s_v[\sigma]=1$}{
                \lIf{$cnt_\sigma<M$}{$cnt_\sigma\leftarrow cnt_\sigma+1$}
                \lElse{$s_v[\sigma]\leftarrow0$}
            }
        }
        
        \If{$s_v\neq0$}{
            $v.status\leftarrow s_v$\;
            $\mathcal{N}(u)\leftarrow\mathcal{N}(u)\cup\{v\}$\;
        }
    }
    
    \KwRet $\mathcal{N}(u),\Delta\mathcal{W}$\;
}
\end{algorithm}

For compactness, we denote the interval of a node $x$ by
$I_x=[l_x,r_x]$. For $\sigma\in\{IF,IS\}$, let $s[\sigma]$ denote
whether semantic $\sigma$ is active in bitmask $s$. We define the semantic
witness conditions as
$\Phi_{IF}(u,v,w): I_w\subseteq I_u\cup I_v$ and
$\Phi_{IS}(u,v,w): I_u\cap I_v\subseteq I_w$, where the IS condition is
considered only when $I_u\cap I_v\neq\emptyset$.

\begin{algorithm}[t]
\caption{UG Interval-Aware Beam Search}
\label{alg:ug_search}
\SetKwInOut{Input}{Input}
\SetKwInOut{Output}{Output}
\SetKwProg{Fn}{Function}{}{end}

\newcommand{\FLAGIF}{\mathsf{FLAG}_{IF}}
\newcommand{\FLAGIS}{\mathsf{FLAG}_{IS}}

\Input{Graph $G$, query $Q=(q_v,q_I,k,type)$, beam size $ef_{search}$}
\Output{Top-$k$ nearest neighbors}

\Fn{\textsc{ContextAwareSearch}($G, Q, ef_{search}$)}{
    $start \leftarrow \textsc{GetEntryNode}(Q.q_I, Q.type)$\;
    \lIf{$start=\text{NULL}$}{\KwRet $\emptyset$}
    
    $C \leftarrow \{start\},\ R \leftarrow \{start\},\ Visited \leftarrow \{start\}$\;
    
    \While{$C \neq \emptyset$}{
        $u \leftarrow \text{ExtractMin}(C)$\;
        
        \If{$|R| \ge ef_{search}$ \textbf{and} $\delta(u,q_v) > \max_{x \in R}\delta(x,q_v)$}{
            \textbf{break}\;
        }
        
        \ForEach{$v \in N(u) \setminus Visited$}{
            $Visited \leftarrow Visited \cup \{v\}$\;
            
            \If{$Q.type=\text{IF}$}{
                \If{$(v.status\ \&\ \FLAGIF)\neq 0$ \textbf{and} $[l_v,r_v] \subseteq Q.q_I$}{
                    $C \leftarrow C \cup \{v\}$\;
                    $R \leftarrow R \cup \{v\}$\;
                    \lIf{$|R| > ef_{search}$}{$\text{RemoveMax}(R)$}
                }
            }
            \ElseIf{$Q.type=\text{IS}$}{
                \If{$(v.status\ \&\ \FLAGIS)\neq 0$ \textbf{and} $Q.q_I \subseteq [l_v,r_v]$}{
                    $C \leftarrow C \cup \{v\}$\;
                    $R \leftarrow R \cup \{v\}$\;
                    \lIf{$|R| > ef_{search}$}{$\text{RemoveMax}(R)$}
                }
            }
        }
    }
    
    \KwRet Top-$k$ nodes from $R$\;
}
\end{algorithm}

Algorithm~\ref{alg:unified_prune} presents the local unified pruning rule used in UG. The algorithm first initializes the neighbor set, the repair set, and the semantic degree counters, and then sorts all candidates by their distance to $u$ (lines~1--4). This ordering allows previously retained neighbors to serve as possible witnesses for later candidates. For each candidate $v$, UG initializes both semantic bits as active, and immediately removes the IS bit if $I_u$ and $I_v$ do not overlap, since no ISANN query can make both endpoints valid in this case (lines~5--8). The pruning step then scans the already retained neighbors of $u$ as candidate witnesses (lines~9--17). If both semantic bits of $v$ have been removed, the scan stops early (lines~10--11). Otherwise, a retained neighbor $w$ is considered only when it satisfies the geometric witness condition $\delta(v,w)<\delta(u,v)$ (lines~12--13); the other RNG-style condition $\delta(u,w)<\delta(u,v)$ is guaranteed by the sorted processing order. For each semantic $\sigma\in\{IF,IS\}$, the bit $s_v[\sigma]$ is cleared when $w$ is also active under $\sigma$ and satisfies the corresponding semantic condition $\Phi_\sigma(u,v,w)$ (lines~14--16). Whenever a semantic bit is removed, the pair $(w,v)$ is recorded in $\Delta\mathcal{W}$ as a repair candidate for the next iteration (line~17). After witness pruning, UG enforces the semantic degree budget separately for IF and IS, so that each semantic projection remains sparse (lines~18--21). Finally, if at least one semantic bit of $v$ remains active, $v$ is retained as a neighbor of $u$ together with its remaining bitmask (lines~22--24). The algorithm returns the pruned neighbor set and the repair set (line~25).

The iterative repair mechanism is introduced to compensate for the bounded candidate pool. In the exact URNG, an edge can be safely pruned because there exists a witness edge together with a monotone continuation path. In practice, however, such a continuation path may not yet be present in the current local candidate pool. UG therefore records the pruned endpoint as a repair candidate of the witness node and reintroduces it in the next round. This makes the iterative construction process a heuristic approximation of monotone-path repair under bounded candidate pools.

\begin{algorithm}[t]
\caption{Entry Node Acquisition}
\label{alg:get_entry}
\SetKwInOut{Input}{Input}
\SetKwInOut{Output}{Output}
\SetKwProg{Fn}{Function}{}{end}

\Input{Query interval $q_I=[q_l,q_r]$, query type $type$, sorted arrays $Arr$}
\Output{Entry node ID or \text{NULL}}

\Fn{\textsc{GetEntryNode}($q_I, type$)}{
    \If{$type=\text{IF}$}{
        $i \leftarrow \text{BinarySearchFirst}(Arr.L,\ge,q_l)$\;
        \If{$i$ is valid \textbf{and} $Arr.SuffMinR[i].val \le q_r$}{
            \KwRet $Arr.SuffMinR[i].node\_id$\;
        }
    }
    \ElseIf{$type=\text{IS}$}{
        $i \leftarrow \text{BinarySearchLast}(Arr.L,\le,q_l)$\;
        \If{$i$ is valid \textbf{and} $Arr.PrefMaxR[i].val \ge q_r$}{
            \KwRet $Arr.PrefMaxR[i].node\_id$\;
        }
    }
    \KwRet \text{NULL}\;
}
\end{algorithm}

\begin{theorem}
\label{thm:asymptotic_heredity}
When $M=\infty$, pruning on the whole candidate graph is equivalent to pruning directly on the candidate graph induced by the query-valid nodes. In other words, the candidate-based pruning procedure preserves heredity asymptotically.
\end{theorem}

\begin{proof}
\comment{Let the candidate graph be $G=(V,E)$. For an interval query $I$, let $G[I]=(V_I,E_I)$ denote the subgraph induced by valid nodes, and let $G'=(V',E')$ denote the graph obtained by directly constructing the candidate graph on the same valid node set. Since $V_I = V'$, the candidate edge sets before pruning are identical.

We then consider all candidate edges in ascending order of distance. The shortest edge cannot be pruned, because there exists no shorter edge that can witness it as the longest edge of a triangle. Now assume inductively that all shorter edges have identical semantic states in both $E_I$ and $E'$. Consider an edge $(u,v)$.

If the IF bit or IS bit of $(u,v)$ remains active in $E_I$, then there does not exist any valid witness edge with the same semantic state that satisfies the corresponding pruning condition. Since the candidate set in $E'$ is a subset of that in $E_I$, such a witness cannot appear in $E'$ either. Hence, the same semantic bit remains active in $E'$.

Conversely, suppose that one semantic bit of $(u,v)$ is inactive in $E_I$. Then there exists a witness edge $(u,w)$ with the same semantic status that satisfies the corresponding pruning condition. For an IFANN query, from the interval containment condition we have
\[
I.l \leq \min(u.l,v.l) \leq w.l \leq w.r \leq \max(u.r,v.r) \leq I.r,
\]
which implies $[w.l,w.r] \subseteq I$, and hence $w \in V_I = V'$. For an ISANN query, from the interval stabbing condition we have
\[
w.l \leq \max(u.l,v.l) \leq I.l \leq I.r \leq \min(u.r,v.r) \leq w.r,
\]
which implies $I \subseteq [w.l,w.r]$,and again $w \in V_I = V'$. Since $\mathrm{dis}(u,w) < \mathrm{dis}(u,v)$, the induction hypothesis ensures that $(u,w)$ has the same semantic state in $E'$. Therefore, $(u,v)$ must also be pruned in $E'$. Thus, the two graphs have identical semantic states for every edge, and hence $E_I = E'$.}

When $M=\infty$, the candidate set is not truncated, so all possible witness nodes are available. Let $G[I]$ be the graph obtained by pruning the whole candidate graph and then restricting it to the query-valid nodes, and let $G'$ be the graph obtained by directly pruning the candidate graph induced by the same valid node set. Since both procedures operate on the same valid node set, they have the same candidate edges before pruning.

We compare edges in ascending order of distance. The shortest edge cannot be pruned. For any edge $(u,v)$, a semantic bit is pruned iff there exists a shorter valid witness node $w$ satisfying the corresponding pruning condition. Such a witness is preserved after restricting to valid nodes: under IF semantics, the pruning condition implies
\[
[w.l,w.r]\subseteq [\min(u.l,v.l),\max(u.r,v.r)]\subseteq I,
\]
so $w$ is query-valid; under IS semantics, it implies
\[
I\subseteq [\max(u.l,v.l),\min(u.r,v.r)]\subseteq [w.l,w.r],
\]
so $w$ is also query-valid. Therefore, every witness used in the whole graph also appears in the induced graph, and vice versa. Hence each edge has the same IF/IS semantic state in $G[I]$ and $G'$, implying $G[I]=G'$.

Therefore, when $M=\infty$, candidate-based pruning preserves structural heredity.
\end{proof}

\begin{theorem}
\label{thm:ug_complexity}
The construction complexity of UG is \(O(TCM_{ug})\), where \(T\) is the number of pruning iterations, \(C=ef_{spatial}+ef_{attribute}\) is the candidate set size, and \(M_{ug}\) denotes the number of effective edge examinations in UG.
\end{theorem}

\begin{proof}
In each iteration, UG examines candidate edges and may generate additional repair candidates from pruned edges. Under a small number of iterations, the candidate pool remains stable up to a constant factor after deduplication. Therefore, the total cost scales linearly with the number of iterations, the candidate size, and the number of effective edge examinations, yielding \(O(TCM_{ug})\).
\end{proof}

\subsection{Query Processing}
\label{subsec:ug_search}

\begin{figure}[t]\vspace*{-1.5em}
  \centering
  \includegraphics[width=\linewidth]{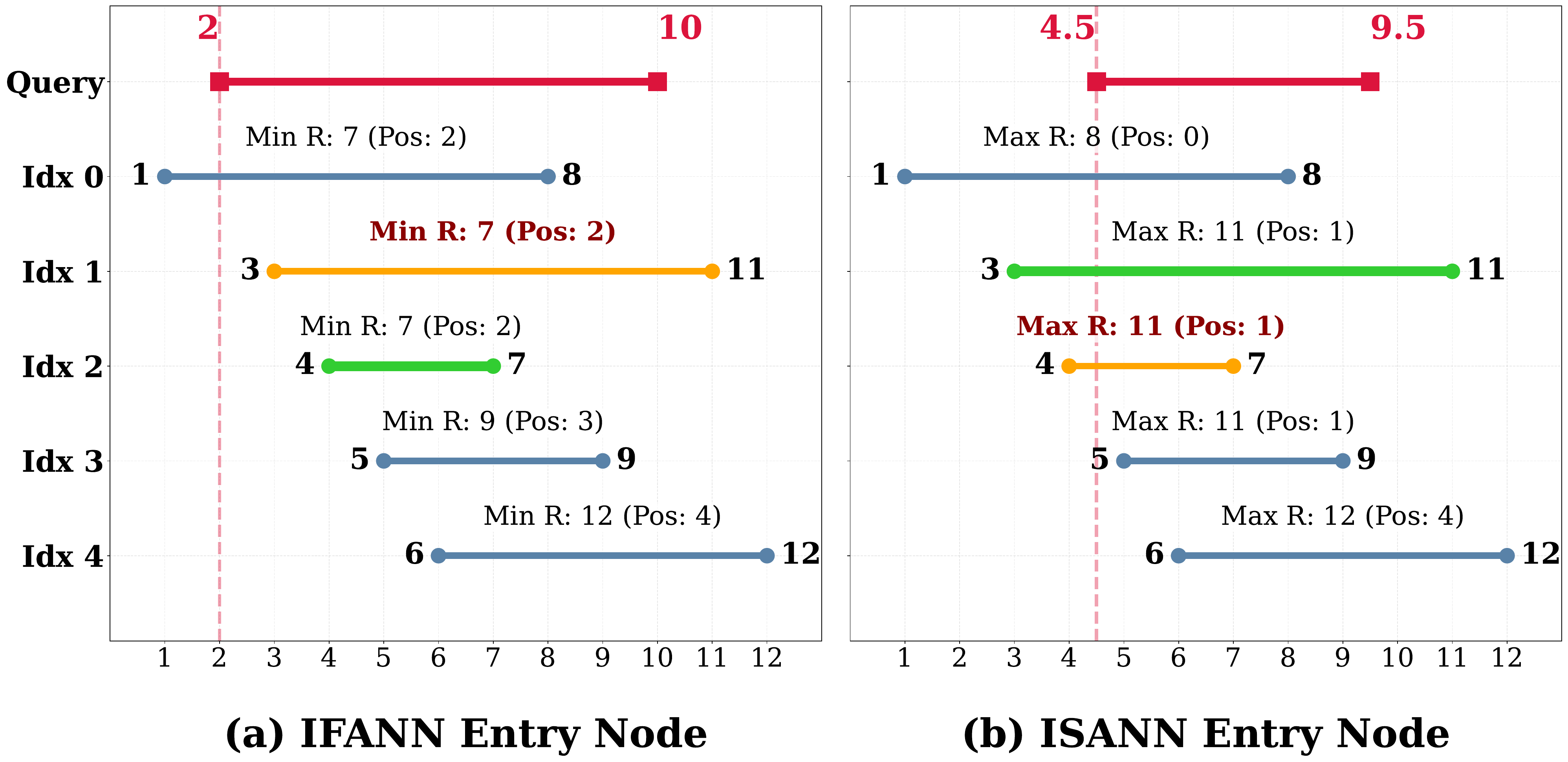}
  \vspace{-2.5em}
  \caption{{Entry Node Acquisition}}
  \vspace{-1.5em}
  \label{fig:ug_entry_node}
\end{figure}

We propose an efficient search algorithm that supports ANN search over the UG. 
% Once UG has been constructed, top-$k$ search is performed by augmenting standard beam search with interval-aware filtering. 
Based on Theorem~\ref{thm:asymptotic_heredity}, UG can be searched directly under the query predicate without explicitly materializing the query-induced subgraph, which would otherwise incur substantial per-query overhead.

To initialize search within the valid subgraph, UG first locates an entry node that already satisfies the query predicate. This is achieved by sorting all nodes by their left endpoints and maintaining two auxiliary arrays: the suffix minimum of right endpoints and the prefix maximum of right endpoints. As shown in Algorithm~\ref{alg:get_entry}, this allows a valid entry node, if one exists, to be found in logarithmic time.

Given the acquired entry node, Algorithm~\ref{alg:ug_search} performs an interval-aware beam search on the unified physical graph. Instead of constructing different graph structures for IFANN and ISANN queries, UG dynamically filters neighbors using the semantic bitmask stored on each edge together with the corresponding interval predicate. In this way, the same physical graph is searched as the logical IFANN-induced or ISANN-induced subgraph, depending on the query type.

\begin{lemma}[Entry-node acquisition]
\label{thm:entry_correctness}
Algorithm~\ref{alg:get_entry} is correct for both IFANN and ISANN queries. In particular:

(1) if the algorithm returns a node \(u \neq \textsc{NULL}\), then \(u\) satisfies the corresponding query predicate;

(2) if the algorithm returns \(\textsc{NULL}\), then no valid node exists in the dataset for that query;

(3) the time complexity of a single invocation is \(O(\log n)\).
\end{lemma}

\begin{proof}
For an IFANN query, the algorithm first finds the smallest index \(i\) such that \(l_u \ge q_l\). Any valid node must therefore lie in the suffix starting from \(i\). If the minimum right endpoint in this suffix is at most \(q_r\), then the returned node satisfies \(l_u \ge q_l\) and \(r_u \le q_r\), and is thus valid. Otherwise, no node in the suffix can satisfy the IFANN predicate, which means no valid IFANN node exists.

For an ISANN query, the algorithm finds the largest index \(i\) such that \(l_u \le q_l\). Any valid node must then lie in the prefix ending at \(i\). If the maximum right endpoint in this prefix is at least \(q_r\), then the returned node satisfies \(l_u \le q_l\) and \(r_u \ge q_r\), and is therefore valid. Otherwise, no node in the prefix can satisfy the ISANN predicate, so no valid ISANN node exists.

In both cases, the algorithm consists of one binary search followed by a constant number of array accesses and comparisons, yielding a total complexity of \(O(\log n)\).
\end{proof}

\paragraph{Example.}
Figure~\ref{fig:ug_entry_node} illustrates the entry-node acquisition procedure for both query types. In Figure~(a), for the IFANN query interval \([2,10]\), the binary search first identifies the suffix of nodes with left endpoints no smaller than \(2\), and the entry node is then selected by the suffix-minimum right endpoint, yielding the valid interval \([4,7]\). In Figure~(b), for the ISANN query interval \([4.5,9.5]\), the binary search identifies the prefix of nodes with left endpoints no greater than \(4.5\), and the entry node is obtained by the prefix-maximum right endpoint, yielding the valid interval \([3,11]\). This example shows that the entry-node acquisition procedure is not only efficient, but also directly aligned with the two interval predicates.

% The example for entry node acquisition will be added later.

\begin{table}[t]
\small
    \centering
    \vspace*{-2em}    
    \caption{{Dataset Statistics}}
    \label{tab:dataset}
    \vspace*{-1.5em}
    \renewcommand{\arraystretch}{1.2}
    \setlength{\arrayrulewidth}{0.8pt}
    \begin{tabularx}{0.45\textwidth}{
        |>{\centering\arraybackslash}p{1.5cm}
        |>{\centering\arraybackslash}X
        |>{\centering\arraybackslash}X   
        |>{\centering\arraybackslash}X        
        |>{\centering\arraybackslash}p{1.5cm}|
        }
        % \toprule
        \hline
        \textbf{Datasets} & $|D|$ & $|Q|$ & $d$ & Vector Type \\
        % \midrule
        \hline %\hline
        % UCF-Crime & 100,000 & 10,000 & 4,096 & video \\
        DB-OpenAI & 990,000 & 10,000 & 1,536 & text \\      
        GIST1M & 1,000,000 & 1,000 & 960 & image \\        
        % MNIST & 60,000 & 10,000 & 784 & image \\        
        % \hline
        S\&P 500 & 1,445,794 & 14,603 & 384 & financial \\
        % GloVe & 1,183,514 & 10,000 & 200 & text \\        
        SIFT1M & 1,000,000 & 10,000 & 128 & image \\  
        % \hline
        DEEP1M & 990,000 & 10,000 & 96 & image \\        
        
        \hline        
        % \bottomrule
    \end{tabularx}
    \vspace*{-2.5em}
\end{table}

% ====================================
% section 5
% ====================================

\section{\yin{Experiments}}
\label{sec:experiments}
% this part is fixed by ziqi yin

\subsection{{Experimental Setup}}
\label{sec:exp-setup}

\noindent\textbf{Datasets.} {We evaluate our methods on five public datasets, including one real-world dataset (S\&P 500) and four widely used benchmark datasets (DB-OpenAI, GIST1M, SIFT1M, and DEEP1M)~\cite{yang2025hi,aumuller2020ann}. The dataset statistics are summarized in Table~\ref{tab:dataset}. Following prior work~\cite{yang2025hi,wang2025timestamp,li2025attribute}, the interval attributes and query ranges are synthetically generated for all datasets except S\&P 500, since such attributes are rarely available due to privacy concerns. For S\&P 500, the attributes are derived from real financial data.
%according to the standard  workload construction protocol when such attributes are not natively available.
}

\begin{figure*}[!t]
\centering
\vspace*{-2em}
\includegraphics[width=\textwidth]{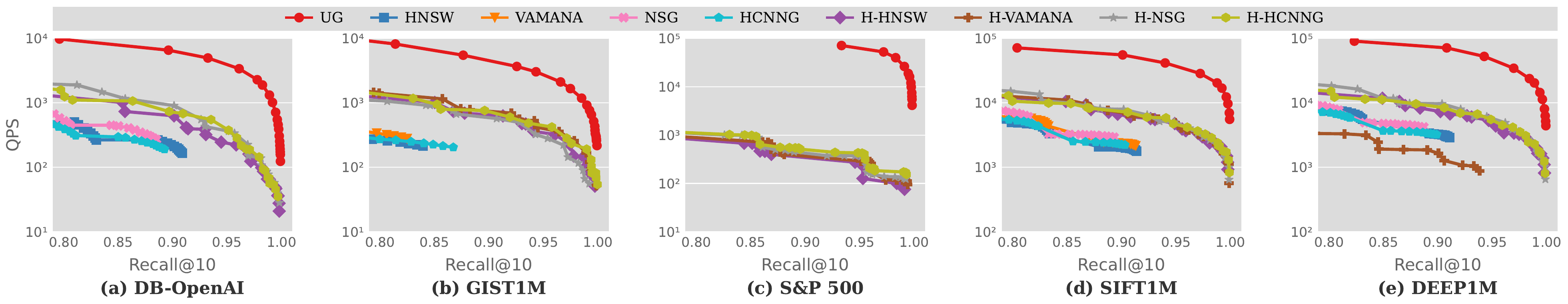}
\vspace*{-2.5em}
\caption{The accuracy–efficiency trade-off for IFANN queries under the uniform workload.}
\vspace*{-1.5em}
\label{fig:exp-1}
\end{figure*}

\begin{figure*}[!t]
  \centering
  \includegraphics[width=\linewidth]{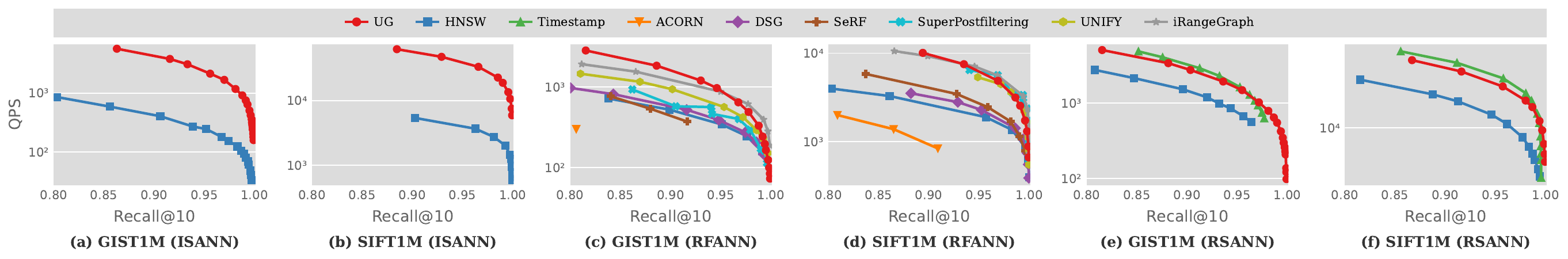}
  \vspace*{-2.5em}
  \caption{{ Query performance under diverse query types}}
  \vspace*{-2em}
  \label{fig:exp-4}
\end{figure*}

%yin{We use six real-world public datasets and their corresponding query sets, which have been widely adopted in prior studies on Interval-Filtering ANN queries~[xx] and Timestamp ANN queries~[xx]. The dataset statistics are summarized in Table~\ref{tab:dataset}. The interval attributes within these datasets are generated using random distributions following the methodology used in~[xx].}

\noindent\textbf{Baselines.} {\underline{First,} for Interval-Filtering ANN (IFANN) queries, we compare against representative methods proposed in~\cite{yang2025hi}, 
%graph-based ANN methods and their hierarchical interval-aware variants, 
including HNSW~\cite{malkov2018efficient}, Vamana~\cite{NEURIPS2019_09853c7f}, NSG~\cite{fu2017fast}, HCNNG~\cite{VARGASMUNOZ2019106970}, Hi-PNG-HNSW (denoted as H-HNSW), Hi-PNG-Vamana (denoted as H-Vamana), Hi-PNG-NSG (denoted as H-NSG), and Hi-PNG-HCNNG (denoted as H-HCNNG). %The plain graph-based baselines 
HNSW, Vamana, NSG, and HCNNG 
adopt a post-filtering strategy, i.e., retrieving top-$k'$ candidates first and then applying interval filtering, while the Hi-PNG variants partition the interval space hierarchically and build graph indexes on the corresponding partitions. \underline{Second,} for Range-Stabbing ANN (RSANN) queries, we compare against Timestamp ~\cite{wang2025timestamp} and HNSW. In particular, HNSW also adopts a post-filtering strategy, while TG constructs timestamp-aware graph indexes and compresses them into a single structure for efficient querying. \underline{Third,} for Range-Filtered ANN (RFANN) queries, we compare against RFANN-specific methods including SeRF~\cite{zuo2024serf}, DSG~\cite{peng2025dynamic}, ACORN~\cite{patel2024acorn}, iRangeGraph~\cite{xu2024irangegraph}, UNIFY~\cite{liang2024unify}, Faiss-HNSW~\cite{douze2025faiss}, SuperPostfiltering~\cite{engels2024approximate}. \underline{Finally,} for Interval-Stabbing ANN (ISANN) queries, we compare against HNSW-hnswlib~\cite{malkov2018efficient} as the only available baseline. Unless otherwise specified, we follow the benchmark implementation and parameter settings in the existing experimental evaluation~\cite{li2025attribute,wang2025timestamp,yang2025hi}.
%published experimental study on attribute-filtering ANN search~\cite{li2025attribute}. 
}

\noindent\textbf{Query Workloads.} {For the IFANN queries, we adopt the real-world interval query data for the S\&P 500 dataset and the uniform workload for the other four datasets. To further analyze robustness under different filtering conditions, we additionally construct three non-uniform workloads on GIST1M, namely short, long, and mixed. Specifically, short queries have selectivity below 5\%, long queries have selectivity above 20\%, and mixed queries contain an equal proportion of short and long intervals. For the experiments on other query types (ISANN, RFANN, and RSANN), we report results on GIST1M and SIFT1M under the uniform workload by default.}

\noindent\textbf{Parameter Settings.}
{\underline{First,} for IFANN queries, the H-PNG builds upon four representative graph-based ANN backbones, including HNSW, Vamana, NSG, and HCNNG. Following previous studies~\cite{yang2025hi}, during the indexing phase, for HNSW, Vamana, and NSG, we set the maximum degree $M=32$ and the construction candidate size $ef_{\text{construction}}=128$. For Vamana, we set $\alpha=1.2$. For HCNNG, we set the number of randomized clustering trees $T=10$ to ensure global connectivity, the leaf size threshold $L_s=1000$ to terminate recursive partitioning, and the local MST degree upper bound $s=5$ to maintain graph sparsity. For the Hi-PNG-specific construction, the leaf size threshold is set to $l_s=10000$, which determines when recursive interval partitioning stops. \underline{Second,} for RSANN queries, we follow the settings used in the original papers for Timestamp and HNSW~\cite{wang2025timestamp}. %, where HNSW adopts the post-filtering strategy.
\underline{Third,} for RFANN queries, we use the settings reported in existing experimental evaluation~\cite{li2025attribute} % on attribute filtering ANN search, 
and follow its configurations for SeRF, DSG, ACORN, iRangeGraph, UNIFY, Faiss-HNSW, and SuperPostfiltering. In particular, we use the same global graph configuration $M=40$ and $ef_{\text{construction}}=1000$ unless otherwise specified. For ACORN, we set $\gamma=25$. For SeRF, we set $M=8$ and $ef_{\max}=ef_{\text{construction}}=1000$. DSG follows the same settings as SeRF. For $\beta$-WST, we set $\textit{split\_factor}=2$ and $\textit{shift\_factor}=0.5$. For iRangeGraph, we use the same global graph configuration. For UNIFY, we set $B=8$, with $\textit{low\_threshold}=5\%$ and $\textit{high\_threshold}=50\%$, and enable the combined filtering strategy. %We cite both this benchmark paper and the corresponding original papers of these baselines.
\underline{Fourth,} for ISANN queries, we use the HNSW-hnswlib %implementation provided by hnswlib 
as the baseline
, with $M=16$ and $ef_{\text{construction}}=200$. For our proposed UG, unless otherwise specified, we use the default setting $ef_{\text{spatial}}=128$, $ef_{\text{attribute}}=300$, $\textit{max\_edges}_{IF}=256$, $\textit{max\_edges}_{IS}=256$, and
5 refinement iterations. %During query processing, we vary the beam size for evaluation and set the number of nearest neighbors to $k=10$. %There are only one exceptions, 
In the scalability experiment, we set $ef_{\text{spatial}}=32$ for nndescent.}

\noindent\textbf{Evaluation Metrics.} {Following previous studies~\cite{zou2026rnsgrangeawaregraphindex,yang2025hi}, we adopt two standard metrics to evaluate querying performance. Accuracy is measured by $\text{recall@}k = \frac{|\mathcal{R}\cap \tilde{\mathcal{R}}|}{k}$, where $\mathcal{R}$ denotes the result set returned by the evaluated method and $\tilde{\mathcal{R}}$ represents the ground-truth set obtained via brute-force search. Consistent with recent studies~\cite{zou2026rnsgrangeawaregraphindex,yang2025hi}, we report recall@10, while also presenting results for varying $k$ values (Exp-6). Efficiency is measured using queries per second (QPS), computed as $\text{QPS} = \frac{|Q|}{t}$, where $|Q|$ queries are processed within time $t$.}

\noindent\textbf{Implementation.} All experiments are conducted on a Linux server equipped with a single AMD Ryzen Threadripper 3990X 64-Core Processor (2.2GHz) and 224 GB of memory. %Except for the RFANN query experiments which use a single thread — because the benchmark code does not support multi-threaded queries and we adopt single-thread evaluation to ensure fairness~\cite{li2025attribute} — all other experiments are conducted using 48 threads. 
All algorithms are implemented in C++ and compiled with {GCC 14.3.0} using the -Ofast optimization flag. 

\subsection{{Experimental Results}}
\label{sec:exp-results}
\noindent{\textbf{Exp-1: Interval-Filtering ANN querying performance.} The query performance of our proposed UG and the compared IFANN baselines under the uniform workload is shown in Figure~\ref{fig:exp-1}. We have the following observations: \textbf{(1) Across all five datasets, UG consistently outperforms other methods.} Specifically, on the real-world S\&P dataset, when recall@10 = 95\%, UG reaches about 62,029 QPS, whereas the strongest baseline, H-HCNNG, achieves only about 587 QPS, yielding a 105$\times$ speedup. \textbf{(2) UG exhibits an even larger advantage in the high-recall regime.} In particular, on the S\&P dataset, at recall@10 = 98\%, UG achieves approximately 40,000 QPS, whereas the strongest baseline, H-NSG, reaches only 136 QPS, corresponding to a 294$\times$ speedup. \textbf{(3) H-HNSW, H-HCNNG, H-NSG, and H-VAMANA exhibit comparable performance across the five datasets, while significantly outperforming their respective counterparts.} This validates that the Hi-PNG algorithm is the key contributor, whereas existing graph-based methods fail to effectively handle such interval-aware ANN queries.}

 % In particular, in the high-recall region, UG consistently outperforms existing hierarchical IFANN methods, including H-HNSW, H-HCNNG, H-NSG, and H-VAMANA. This result shows that UG can substantially improve the efficiency of interval-filtered ANN queries on real-world data distributions. Moreover, UG also maintains a clear advantage on the ultra-high-dimensional DB-OpenAI dataset. In the high-accuracy region around recall@10 = 95\%, UG delivers about 4,092 QPS, while the strongest H-series baseline achieves only about 376.45 QPS, corresponding to a speedup of about 10.87$\times$. This indicates that the graph structure of UG remains robust even for ultra-high-dimensional vectors. }
\begin{figure*}[!t]
\centering
\vspace*{-2em}
\includegraphics[width=\textwidth]{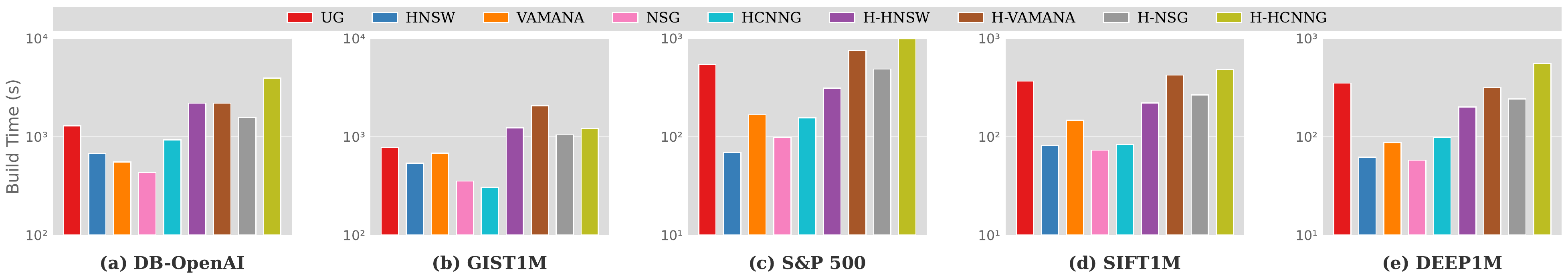}
\vspace*{-2.5em}
\caption{{The indexing time of all IFANN baselines across the five datasets.}}
\vspace*{-1.5em}
\label{fig:exp-2-1}
\end{figure*}

\begin{figure*}[!t]
\centering
\includegraphics[width=\textwidth]{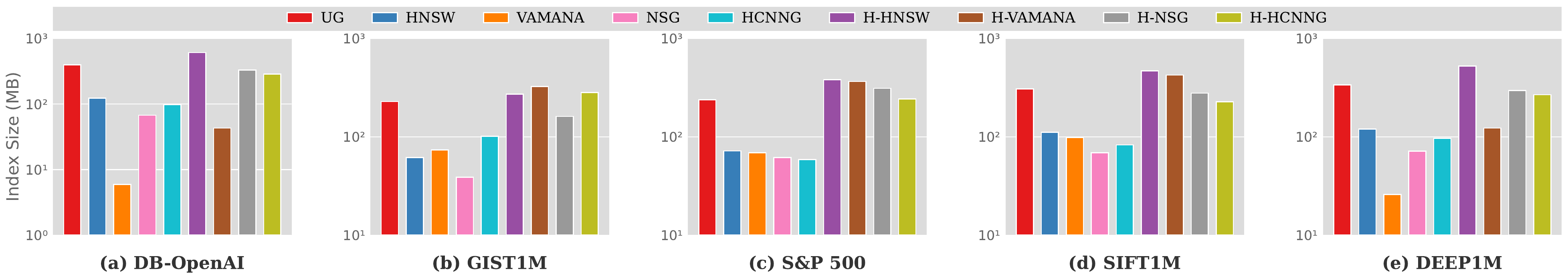}
\vspace*{-2.5em}
\caption{{The memory cost of all IFANN baselines across the five datasets.}}
\vspace*{-1.5em}
\label{fig:exp-2-2}
\end{figure*}

\begin{figure}[t!]
  \centering
  \includegraphics[width=\linewidth]{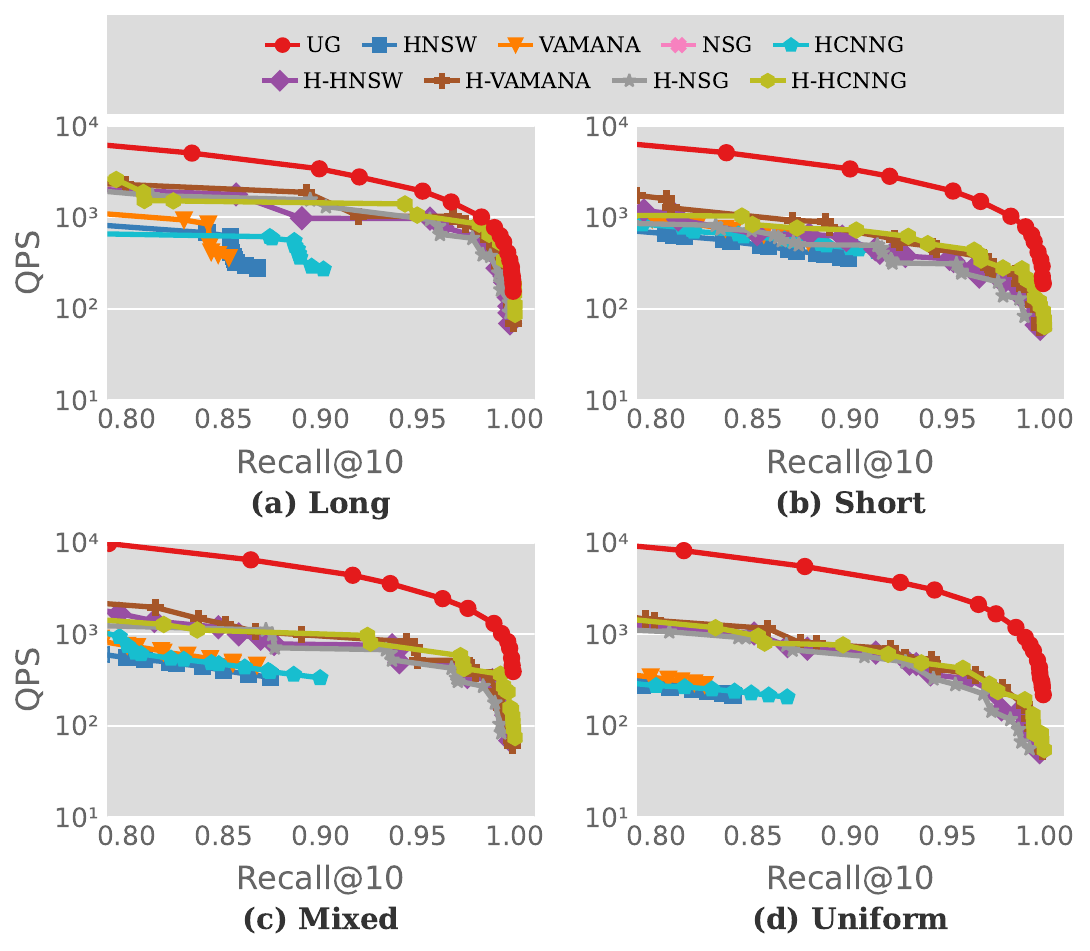}
  \vspace*{-2.5em}
  \caption{{The accuracy–efficiency trade-off for IFANN queries under diverse filtering workloads on GIST1M.}}
  \vspace*{-2.5em}
  \label{fig:exp-3}
\end{figure}

\noindent{\textbf{Exp-2: Query performance under diverse query types.} The query performance of our proposed UG and the ISANN, RFANN, and RSANN baselines is shown in Figure~\ref{fig:exp-4}. \textbf{The experimental results demonstrate that our proposed UG consistently delivers state-of-the-art query performance.} 
%This section compares the QPS--Recall performance of UG on . Overall, the behavior of UG varies across query types.
In particular, for ISANN, UG consistently outperforms the baseline. Under the uniform workload, when recall@10 = 95\%, UG achieves about 2127.66 QPS on GIST1M and 33333 QPS on SIFT1M, whereas HNSW achieves only about 245 and 3676 QPS, corresponding to speedups of about 8.7$\times$ and 9.1$\times$, respectively. For RSANN and RFANN, UG also achieves performance comparable to that of query-specific optimized baselines. It is worth noting that we investigate the scenario in which queries of different types are issued in the same query workload, which suggests that UG provides the best overall solution for handling diverse query types in diverse query workloads.}

\noindent{\textbf{Exp-3: Query performance under diverse filtering workloads.} To further assess the benefits of the proposed UG across different query workloads, we evaluate the IFANN query performance of all methods on the Long, Mixed, Short, and Uniform workloads of the GIST1M dataset. The querying performance is presented in Figure~\ref{fig:exp-3}. \textbf{The results confirm that our proposed UG consistently maintains its advantage across different workloads.} In particular, when recall@10 = 95\%, UG achieves about 1,999, 2,975, 2,014, and 2,761 QPS under the Long, Mixed, Short, and Uniform workloads, respectively. In comparison, the fastest baselines only achieve about 1,056, 684, 482, and 445 QPS under these four workloads, respectively, leading to speedups of about 1.9$\times$, 4.4$\times$, 4.2$\times$, and 6.2$\times$. This indicates that UG can sustain substantially higher throughput even when the filtering condition becomes more selective.}

%impact of filtering distributions on ,   and present the corresponding QPS--Recall curves in a figure in More specifically, in the high-recall region, UG consistently outperforms existing hierarchical IFANN baselines under all four workloads. Around }

%does not surpass the task-specific SOTA methods, which is expected. Around recall@10 = 95\%, the gap between UG and Timestamp on RSANN is about 1.24$\times$--1.89$\times$, while on RFANN the gap between UG and the strongest specialized methods (e.g., iRangeGraph or SuperPostfiltering) is about 2.88$\times$--3.11$\times$.We believe this reflects a trade-off between generality and specialization. Existing SOTA methods for RFANN and RSANN are usually optimized aggressively for a single query type, whereas UG is designed to support IF, IS, and other more general query patterns within a unified framework. As a result, UG is not always optimal on highly specialized tasks, but it still remains competitive, especially on RSANN where the gap is moderate. This suggests that the main value of UG is not to dominate every specialized setting individually, but to provide stronger unified query support at an acceptable performance cost.}

\noindent{\textbf{Exp-4: Indexing Cost.} Figures~\ref{fig:exp-2-1} and~\ref{fig:exp-2-2} report the indexing time and index size of different IFANN methods on the five datasets, while similar patterns are observed on other query types. \textbf{In terms of indexing time, UG is comparable to competitive baselines and even outperforms them on several datasets.} 
%is highly competitive and even 
For example, on GIST1M, UG takes about 678.9 s to build the index, whereas the best baseline, H-NSG, requires about 947.5 s, yielding a 1.4$\times$ speedup. %On DB-OpenAI, UG is also faster than all H-series methods. 
This shows that UG not only delivers superior query performance but also offers a clear advantage in index construction time. 
\textbf{For memory consumption, UG maintains a moderate index size compared with competitive baselines.} For instance, on GIST1M, UG uses about 218.7 MB, substantially less than H-HCNNG and H-HNSW, which require 272.4 MB and 262.0 MB. Equivalently, the indexes of H-HCNNG and H-HNSW are about 1.25$\times$ and 1.20$\times$ larger than that of UG, respectively. These results indicate that UG does not achieve its query advantage at the cost of excessive memory consumption. Overall, UG provides a favorable trade-off between indexing overhead and query efficiency.}

\noindent{\textbf{Exp-5: Querying performance with different $k$.} We further analyze the effect of different $k$ values on query performance on the GIST1M dataset, %under the uniform workload, 
as shown in Figure~\ref{fig:exp-6}. \textbf{Overall, QPS gradually decreases as k increases steadily.}%for both IFANN and RSANN queries, QPS gradually decreases as k increases, while the attainable recall improves. 
This is expected: a larger $k$ requires the system to return more results, which expands the search scope and increases the cost of candidate exploration. %At the same time, to maintain result quality under larger $k$, the search process must examine a sufficiently rich candidate set.
These results demonstrate that our proposed UG adapts well to varying retrieval requirements. %Although larger k inevitably introduces additional query cost, the reduction in QPS remains moderate, which further validates the efficiency and effectiveness of our proposed solution across different scenarios.
}

\noindent\textbf{Exp-6: Sensitivity study.} We conduct a parameter sensitivity study for UG on IFANN queries using the GIST1M dataset %under the uniform workload 
and the results are shown in Figure~\ref{fig:exp-5}. \textbf{This proves that UG is not highly sensitive to parameter variations.} In particular, %it maintains relatively stable QPS--Recall curves over a broad range of settings, indicating its robustness. 
for $ef_{attribute}$ and $ef_{spatial}$, increasing the parameter generally improves the high-recall region, although the gain gradually saturates. %From the overall trend of the curves, larger values tend to push the curves toward higher recall, while the differences across nearby settings are often moderate, showing that UG is fairly stable with respect to these search-width parameters. 
Meanwhile, larger values also lead to longer index construction time and higher memory consumption, so overly aggressive settings are unnecessary. The effect of iteration is more direct, and most of the gain is obtained within the first few iterations. As the iteration count increases from 1 to 3, the QPS-Recall curve improves substantially. However, the additional gain becomes small when it further increases to 4 and 5. % In terms of maximum recall, it rises from about \yin{97.60\%} at iteration=1 to about \yin{99.30\%} at iteration=2 and \yin{99.52\%} at iteration=3, and then only slightly improves to about 0.9965 and 0.9968 at iteration=4 and iteration=5, respectively. 
This shows that only a small number of iterations is needed to obtain most of the benefit. In contrast, %$max_{edges}$ has a more pronounced impact on the high-recall region. %Overly 
small values of $max_{edges}$ noticeably limit the final performance, while medium values already capture most of the gain. %For instance, when $max_{edges}$ increases from 75 to 200, the maximum recall improves from 97.44\% to 99.83\%, showing that a moderate increase in graph connectivity is effective for improving high-accuracy query performance; 
However, the gain also gradually saturates as the parameter becomes larger. At the same time, larger $max_{edges}$ values also imply higher construction and memory overhead. %Overall, these results suggest that UG is not highly sensitive to parameter variations. %does not rely on extreme parameter settings. 
In practice, moderate parameter values are already sufficient to achieve a favorable balance between query performance and indexing overhead.

\begin{figure}[t]\vspace*{-2em}
  \centering
  \includegraphics[width=\linewidth]{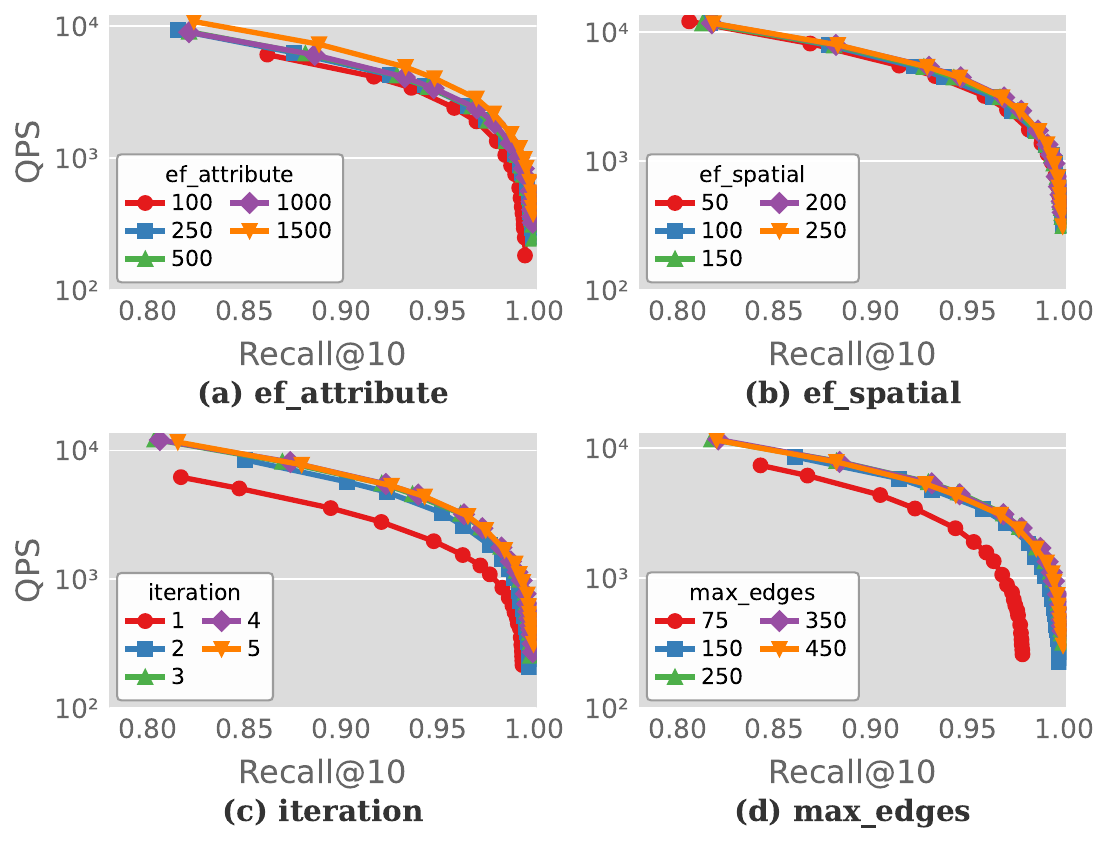}
  \vspace{-2.5em}
  \caption{{Parameter sensitivity}}
  \vspace{0em}
  \label{fig:exp-5}
\end{figure}

\noindent{\textbf{Exp-7: Scalability Study.} To evaluate the scalability of UG, we conduct IFANN experiments with increasing dataset size 
%on four subsets extracted 
from SIFT100M, namely 10M, 20M, 30M, and 40M. Figure~\ref{fig:exp-7} reports the query latency at recall@10 = 90\% and the index build time under different dataset sizes. Overall, as the dataset size increases, both query latency and build time grow steadily, while the overall trend remains smooth, demonstrating good scalability of UG on larger datasets.}

\begin{figure}[t]\vspace*{-1.5em}
  \centering
  \includegraphics[width=\linewidth]{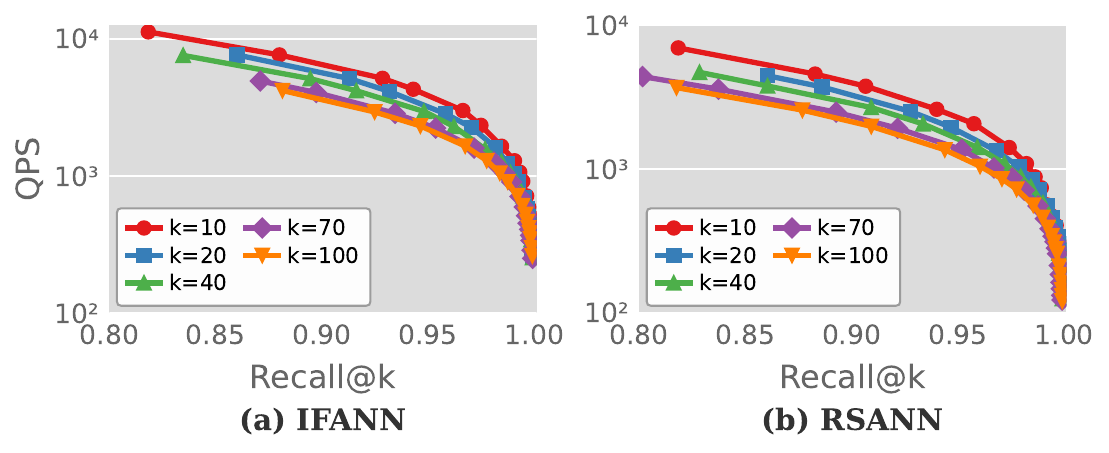}
  \vspace{-2.5em}
  \caption{{Querying performance with different $k$.}}
  \vspace{-2.0em}
  \label{fig:exp-6}
\end{figure}

\section{{Related Work}}

\noindent{\textbf{Approximate Nearest Neighbor Search.} Existing ANN indexes can be broadly categorized into four categories: graph-based~\cite{azizi2025graph,fu2017fast,harwood2016fanng,malkov2014approximate,malkov2018efficient,wang2021comprehensive}, quantization-based~\cite{gao2024rabitq,ge2013optimized,gong2012iterative,jegou2010product,matsui2018survey}, hashing-based~\cite{pham2022falconn++,wang2014hashing,wang2017survey}, and tree-based~\cite{beygelzimer2006cover,ram2019revisiting}. A more overview is provided in recent surveys ~\cite{pan2023survey,pan2024vector}. Extensive empirical evaluations~\cite{aumuller2020ann,azizi2025graph,li2019approximate,wang2017survey} have consistently shown that graph-based methods achieve state-of-the-art performance. This superiority stems from the graph's ability to encode proximity relationships, allowing query processing to converge rapidly by evaluating only a small fraction of the dataset~\cite{dobson2023scaling}. Most graph-based indexes ~\cite{wang2021comprehensive} are built upon four fundamental graph types: the Delaunay Graph (DG)~\cite{fortune2017voronoi}, Relative Neighborhood Graph ({RNG})~\cite{toussaint1980relative}, K-Nearest Neighbor Graph (KNNG)~\cite{paredes2005using}, and Minimum Spanning Tree (MST)~\cite{kruskal1956shortest}. Among these, {RNG}-based indexes achieve particularly strong performance due to their effective pruning strategy~\cite{wang2021comprehensive}. Building on {RNG}, we propose a novel %RNG
{ANN} index %for
to support efficient Interval‑aware ANN search.}

\noindent{\textbf{Attribute-filtered ANN Search.} Existing attribute-filtered ANN queries typically involve two types of attributes: keyword and numerical. Keyword attributes are unordered, whereas numerical attributes are inherently ordered. Consequently, different methods have been proposed to handle different types of attributes~\cite{patel2024acorn,zuo2024serf,yin2024listlearningindexspatiotextual}. For numerical attributes, several specialized queries have been proposed ~\cite{zuo2024serf,yang2025hi}, including RFANN and IFANN queries, as described in Section 2. To support these queries, various indexes have been introduced~\cite{wang2021milvus,xu2024irangegraph,yang2025hi}. Early studies~\cite{douze2025faiss,wang2021milvus,zhang2025efficient} follow a two-stage paradigm, constructing separate indexes for vector and attribute retrieval before merging the results. For example, ADBV~\cite{wei2020analyticdb} employs a B-Tree and a PQ index independently, selecting a retrieval strategy via a cost model. A key limitation of such approaches is that they do not co-optimize the index structure, resulting in sub-optimal performance compared with hybrid indexes ~\cite{chronis2025filtered,engels2024approximate,li2025attribute,liang2024unify,xu2024irangegraph,zuo2024serf}. Recently, many hybrid indexes have been proposed, such as iRangeGraph ~\cite{xu2024irangegraph}, which integrates graph and attribute indexes more directly by constructing a graph for each segment of a partitioned attribute range. While this design achieves high performance, it incurs substantial indexing overhead and exhibits limited adaptability to different query types. For keyword attributes, similar techniques and strategies have been proposed, such as Faiss~\cite{douze2025faiss}, Milvus~\cite{wang2021milvus}, VBASE~\cite{zhang2023vbase}, ACORN~\cite{patel2024acorn}, Filtered-DiskANN~\cite{gollapudi2023filtered} (microsoft). However, they also suffer from either low query performance or high indexing time and limited adaptability to different query types. Recently, several efforts have aimed to design a unified index to support different query types. For example, \cite{xie2025beyond} supports both conventional ANN queries and ANN queries with keyword predicates within a single index. However, designing one index to support for diverse numerical-attribute queries remains underexplored.}

\begin{figure}[t]
  \vspace*{-2em}
  \centering
  \includegraphics[width=\linewidth]{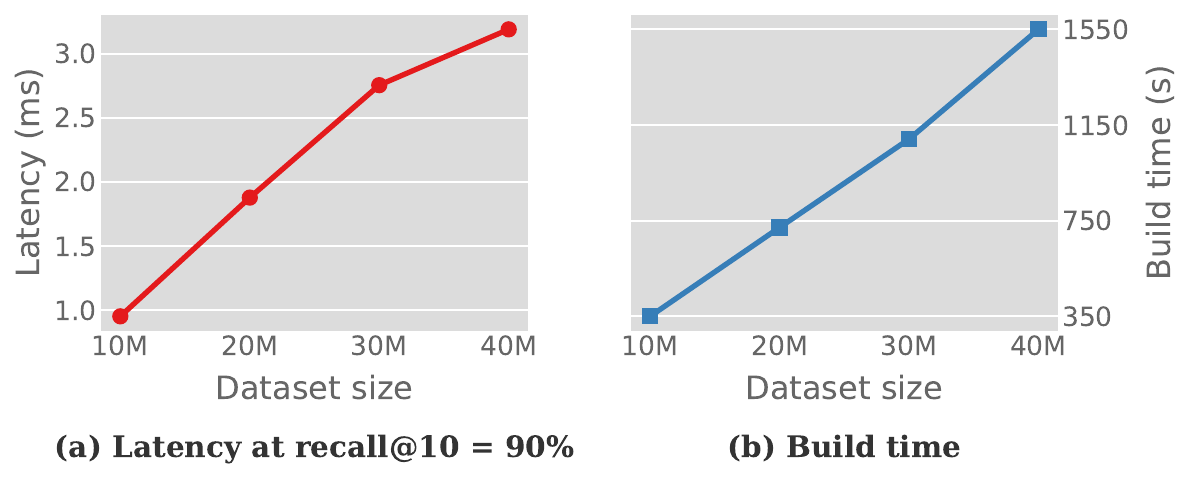}
  \vspace{-2.5em}
  \caption{{Scalability with increasing dataset size.}}
  \vspace{-1.5em}
  \label{fig:exp-7}
\end{figure}

\section{{Conclusion and Future Work}}
{In this paper, we propose an efficient graph index, UG, to handle various interval-aware ANN queries within a single index. We first introduce the Unified Interval-aware Relative Neighborhood Graph (URNG), which establishes the theoretical foundation for interval-aware ANN search. We then present UG, an efficient and practical approximation of URNG. Experimental results show that, for various interval-aware ANN queries, UG significantly outperforms the baselines in search efficiency, while incurring no additional index construction time or storage overhead.}

{One direction for future work is to extend our theoretical graph and index beyond numerical attributes to support more diverse filtered search scenarios, such as keyword filtering. Building upon this, we aim to investigate and develop a %comprehensive, 
unified theoretical framework capable of seamlessly accommodating arbitrary predicate filtering conditions, ultimately providing a generalized index %routing structure 
for complex hybrid queries.}
%unifies interval-containment-related queries within a single graph-based indexing framework. We first introduce the \emph{Unified Interval-aware Relative Neighborhood Graph} (URNG), which provides the theoretical foundation for attribute-filtered search. URNG guarantees the retrieval of the exact nearest neighbor under beam search and achieves an expected theoretical complexity on the same order as the MRNG built without attributes. 
% Since the exact construction of URNG is prohibitively expensive in practice, 
% In addition, our study suggests that the proposed pruning-aware condition can be extended beyond numerical attributes to support more general filtered search scenarios. Exploring such generalizations constitutes an important direction for future work.

\bibliographystyle{ACM-Reference-Format}
\bibliography{UG-used}

\appendix
\section{Proofs for the Complexity Analysis}
\label{app:complexity_proofs}

\begin{lemma}[Expected length of a monotonic path]
\label{lem:monotonic_path_length}
For any starting node \(x \in V\) and the interval-aware nearest neighbor \(y \in V\) of a query \(q\), the expected number of steps of a greedy monotonic walk from \(x\) to \(y\) on \(G\) is \(O(n^{1/d}\log n / \Delta r)\), where \(n:=|V|\) and \(\Delta r\) denotes the minimum gap between edge lengths among all non-isosceles triangles in the metric space.
\end{lemma}

\begin{proof}
Since URNG satisfies monotonic searchability, it belongs to the class of monotonic search networks (MSNETs) in the NSG literature. Therefore, the expected greedy path length on URNG is upper bounded by the corresponding MSNET bound, namely \(O(n^{1/d}\log(n^{1/d})/\Delta r)\), which is \(O(n^{1/d}\log n/\Delta r)\).
\end{proof}

\begin{lemma}[Average degree bound]
\label{lem:urng_degree_bound}
Let \(D_r\) denote the average degree of RNG. Then the average degree of URNG is bounded by \(O(C_{\mathrm{urng}}D_r)\). Under the uniform interval model, \(C_{\mathrm{urng}}\) is a constant independent of \(n\) and may take \(C_{\mathrm{urng}}=6+\frac{13}{3}=\frac{31}{3}\).
\end{lemma}

\begin{proof}
Fix a node \(u \in V\), and partition the unit sphere centered at \(u\) into \(N_d\) cones of half-angle \(30^\circ\). For a fixed cone, order the candidate points by distance to \(u\), let \(Y_t\) indicate whether the \(t\)-th point is retained, and let \(K=\sum_t Y_t\). Then \(\mathbb E[K]=\sum_t \Pr(Y_t=1)\).

For IF semantics, let \(J_t=I_u\cup I_t\) and let \(D=|J_t|\). Conditioned on \(D=p\), a random interval is fully contained in \(J_t\) with probability \(p^2\), so \(\Pr(Y_t=1\mid D=p)=(1-p^2)^{t-1}\). Since \(D\) is the range of four i.i.d. uniform variables, its density is \(f_D(p)=12p^2(1-p)\) on \((0,1)\). Hence \(\mathbb E[K]=\int_0^1 \frac{1-(1-p^2)^m}{p^2}\,12p^2(1-p)\,dp = 12\int_0^1 (1-p)(1-(1-p^2)^m)\,dp\), which converges to \(6\) as \(m\to\infty\).

For IS semantics, let \(J_t=I_u\cap I_t=[L,R]\), which may be empty. Conditioned on \(J_t=[L,R]\neq\varnothing\), a random interval covers it with probability \(2L(1-R)\), so \(\Pr(Y_t=1\mid J_t=[L,R])=(1-2L(1-R))^{t-1}\). Moreover, \(\Pr(J_t=\varnothing)=1/3\), and when \(J_t\neq\varnothing\), its density is \(f_J(L,R)=16L(1-R)\) on \(0<L<R<1\). Therefore, \(\mathbb E[K]=\frac13+\int_{0<L<R<1}16L(1-R)\sum_{t=1}^m (1-2L(1-R))^{t-1}\,dL\,dR\), which simplifies to \(\frac13+8\int_{0<L<R<1}(1-(1-2L(1-R))^m)\,dL\,dR\), and converges to \(\frac{13}{3}\) as \(m\to\infty\).

Thus the expected number of retained neighbors per cone is bounded by \(C_{\mathrm{urng}}=6+\frac{13}{3}\). Summing over all cones shows that the average degree of URNG is bounded by a constant factor times that of RNG, namely \(O(C_{\mathrm{urng}}D_r)\).
\end{proof}

\begin{proof}[Proof of Theorem~\ref{thm:urng_search_complexity}]
By Lemma~\ref{lem:monotonic_path_length}, the expected number of expansion steps is \(O(n^{1/d}\log n/\Delta r)\). By Lemma~\ref{lem:urng_degree_bound}, each expansion examines \(O(C_{\mathrm{urng}}D_r)\) neighbors in expectation. Therefore, the total expected search complexity is \(O(C_{\mathrm{urng}}D_r\,n^{1/d}\log n/\Delta r)\). Since \(C_{\mathrm{urng}}\) is constant and \(D_r\) is bounded independently of \(n\), this reduces to \(O(n^{1/d}\log n/\Delta r)\). Under the spacing assumption \(\Delta r \approx O(n^{-\epsilon/d})\) with \(0<\epsilon<d\), the bound becomes \(O(n^{(1+\epsilon)/d}\log n)\), which is well approximated by \(O(n^{1/d}\log n)\) when \(\epsilon \ll d\).
\end{proof}

\begin{proof}[Proof of Theorem~\ref{thm:urng_index_size}]
The total index size consists of the storage for nodes and edges. The number of nodes is \(O(n)\). By Lemma~\ref{lem:urng_degree_bound}, the expected degree upper bound of URNG is \(O(C_{\mathrm{urng}}D_r)\), where \(D_r\) is the average bound of RNG. Since the total number of edges in RNG is \(S_r=nD_r\), the expected total number of edges in URNG is \(O(C_{\mathrm{urng}}nD_r)=O(C_{\mathrm{urng}}S_r)\). Because \(C_{\mathrm{urng}}\) is independent of \(n\), the result follows.
\end{proof}

\begin{proof}[Proof of Theorem~\ref{thm:exact_build}]
In the exact construction procedure, one enumerates each node \(x\), sorts the remaining nodes by their distances to \(x\), and then checks all candidate witness nodes to determine whether each edge can be pruned. The sorting stage costs \(O(n^2d)\), while the witness enumeration stage costs \(O(n^3)\). Since typically \(n \gg d\), the overall complexity is \(O(n^3)\).
\end{proof}

\end{document}